\documentclass{amsart}
\usepackage{hyperref}
\usepackage{graphicx}
\usepackage{curves}

\newcommand{\arxiv}[1]{\href{http://arxiv.org/abs/#1}{arXiv:#1}}
\newcommand*{\mailto}[1]{\href{mailto:#1}{\nolinkurl{#1}}}

\newtheorem{theorem}{Theorem}[section]
\newtheorem{lemma}[theorem]{Lemma}
\newtheorem{corollary}[theorem]{Corollary}
\newtheorem{remark}[theorem]{Remark}
\newtheorem{hypothesis}[theorem]{Hypothesis {\bf H.}\hspace*{-0.6ex}}

\newcommand{\R}{{\mathbb R}}
\newcommand{\N}{{\mathbb N}}
\newcommand{\Z}{{\mathbb Z}}
\newcommand{\C}{{\mathbb C}}
\newcommand{\M}{{\mathbb M}}

\newcommand{\nn}{\nonumber}
\newcommand{\be}{\begin{equation}}
\newcommand{\ee}{\end{equation}}
\newcommand{\bea}{\begin{eqnarray}}
\newcommand{\eea}{\end{eqnarray}}
\newcommand{\ul}{\underline}
\newcommand{\ol}{\overline}
\newcommand{\ti}{\tilde}

\newcommand{\spr}[2]{\langle #1 , #2 \rangle}
\newcommand{\id}{\mathbb{I}}
\newcommand{\I}{\mathrm{i}}
\newcommand{\E}{\mathrm{e}}
\newcommand{\ind}{\mathrm{ind}}
\newcommand{\re}{\mathrm{Re}}
\newcommand{\im}{\mathrm{Im}}

\DeclareMathOperator{\res}{Res}

\def\XXint#1#2#3{{\setbox0=\hbox{$#1{#2#3}{\int}$}
     \vcenter{\hbox{$#2#3$}}\kern-.5\wd0}}

\newcommand{\sigI}{\begin{pmatrix} 0 & 1 \\ 1 & 0 \end{pmatrix}}

\newcommand{\rN}{\begin{pmatrix}  0 & 0 \end{pmatrix}}

\newcommand{\ulz}{\underline{z}}
\newcommand{\di}{\mathcal{D}}
\newcommand{\vrc}{\ul{\Xi}_{E_0}}

\newcommand{\hmu}{\hat{\mu}}
\newcommand{\uhmu}{{\underline{\hat{\mu}}(n,t)}}
\newcommand{\uhnuz}{{\underline{\hat{\nu}}}}
\newcommand{\uhnu}{{\underline{\hat{\nu}}(n,t)}}
\newcommand{\uhmuz}{{\underline{\hat{\mu}}}}
\newcommand{\dimu}[1]{\di_{\ul{\hat{\mu}}(#1)}}
\newcommand{\dimus}[1]{\di_{\ul{\hat{\mu}}(#1)^*}}
\newcommand{\dimuz}{\di_{\ul{\hat{\mu}}}}
\newcommand{\dimuzs}{\di_{\ul{\hat{\mu}}^*}}
\newcommand{\dinu}[1]{\di_{\ul{\hat{\nu}}(#1)}}
\newcommand{\dinus}[1]{\di_{\ul{\hat{\nu}}(#1)^*}}
\newcommand{\dinuz}{\di_{\ul{\hat{\nu}}}}
\newcommand{\dinuzs}{\di_{\ul{\hat{\nu}}^*}}
\newcommand{\dirho}{\di_{\ul{\rho}}}
\newcommand{\dirhos}{\di_{\ul{\rho}^*}}
\newcommand{\Amap}{\ul{A}_{E_0}}
\newcommand{\amap}{\ul{\alpha}_{E_0}}

\newcommand{\Rg}[1]{R_{2g+2}^{1/2}(#1)}

\newcommand{\eps}{\varepsilon}

\newcommand{\sig}{\sigma}
\newcommand{\lam}{\lambda}
\newcommand{\gam}{\gamma}
\newcommand{\om}{\omega}

\numberwithin{equation}{section}

\begin{document}

\title[Stability of the Periodic Toda Lattice]{Stability of the Periodic Toda Lattice in the Soliton Region}

\author[H. Kr\"uger]{Helge Kr\"uger}
\address{Department of Mathematics\\ Rice University\\ Houston\\ TX 77005\\ USA}
\email{\mailto{helge.krueger@rice.edu}}
\urladdr{\url{http://math.rice.edu/~hk7/}}

\author[G. Teschl]{Gerald Teschl}
\address{Faculty of Mathematics\\
Nordbergstrasse 15\\ 1090 Wien\\ Austria\\ and International Erwin Schr\"odinger
Institute for Mathematical Physics\\ Boltzmanngasse 9\\ 1090 Wien\\ Austria}
\email{\mailto{Gerald.Teschl@univie.ac.at}}
\urladdr{\url{http://www.mat.univie.ac.at/~gerald/}}

\thanks{Research supported by the Austrian Science Fund (FWF) under Grant No.\ Y330.}
\thanks{Int. Math. Res. Not. Int. {\bf 2009}, Art. ID rnp077, 36 pp (2009)}


\date{\today}

\begin{abstract}
We apply the method of nonlinear steepest descent to compute the long-time asymptotics of the
periodic (and slightly more generally of the quasi-periodic finite-gap) Toda lattice for decaying
initial data in the soliton region. In addition, we show how to reduce the problem in the remaining
region to the known case without solitons.
\end{abstract}

\maketitle

\section{Introduction}

Consider the doubly infinite Toda lattice in Flaschka's
variables (see e.g.\ \cite{tjac}, \cite{taet}, or \cite{ta})
\be \label{TLpert}
\aligned
\dot b(n,t) &= 2(a(n,t)^2 -a(n-1,t)^2),\\
\dot a(n,t) &= a(n,t) (b(n+1,t) -b(n,t)),
\endaligned
\ee
$(n,t) \in \Z \times \R$,
where the dot denotes differentiation with respect to time. We
will consider a quasi-periodic finite-gap background
solution $(a_q,b_q)$, to be described in the next section, plus
a short range perturbation $(a,b)$ satisfying
\be \label{decay}
\sum_n (1+ |n|)^{3+l} (|a(n,t) - a_q(n,t)| + |b(n,t)- b_q(n,t)|) < \infty.
\ee
for some $l\in\N$. It suffices to check this condition
for one $t\in\R$ (see \cite{emt2}).
The perturbed solution can be computed via the inverse
scattering transform. The case where $(a_q,b_q)$ is constant
is classical (see again \cite{tjac} or \cite{ta}) and the more
general case we want to apply here was solved only recently in \cite{emt2}
(see also \cite{mt}). The long-time asymptotics in the case where $(a_q,b_q)$ is constant
were first computed by Novokshenov and Habibullin \cite{nh} and later made
rigorous by Kamvissis \cite{km} under the additional assumption that no solitons
are present. The case of solitons was recently investigated by us in \cite{krt}.
For a self-contained introduction, including further references and a more detailed
history of this problem,  see our review \cite{krt2}. The long-time asymptotic in the present
situation were first established by Kamvissis and Teschl in \cite{kt2} (see also \cite{kt} for a short overview)
by a generalization of the so-called nonlinear stationary phase/steepest descent method for oscillatory
Riemann--Hilbert problem deformations to Riemann surfaces. While \cite{kt2} contains only leading
order asymptotics, higher asymptotics were given in \cite{kt3}. However, both \cite{kt2} and \cite{kt3}
assume that no solitons are present and hence the purpose of the present paper is to show how the results
can be extended to cover solitons as well.

To fix our background solution, choose a Riemann surface $\M$ as in \eqref{eq:r12},
a Dirichlet divisor $\dimuz$, and introduce
\be
\ulz(n,t) = \Amap(\infty_+) - \amap(\dimuz) - n\ul{A}_{\infty_-}(\infty_+)
+ t\ul{U}_0 - \vrc \in \C^g,
\ee
where $\Amap$ ($\amap$) is Abel's map (for divisors) and $\vrc$, $\ul{U}_0$ are 
some constants defined in Section~\ref{secAG}. Then our background solution is given in terms
of Riemann theta functions by
\begin{align} \nn
a_q(n,t)^2 &= \ti{a}^2 \frac{\theta(\ulz(n+1,t)) \theta(\ulz(n-1,t))}{\theta(
\ulz(n,t))^2},\\
b_q(n,t) &= \tilde{b} + \frac{1}{2}
\frac{d}{dt} \log\Big(\frac{\theta(\ulz(n,t)) }{\theta(\ulz(n-1,t))}\Big),
\end{align}
where $\ti{a}$, $\tilde{b}$ are again some constants. We remark that this class contains all periodic solutions
as a special case.

In order to state our main result, we begin by recalling that the sequences $a(n,t)$, $b(n,t)$, $n\in\Z$,
for fixed $t\in\R$, are uniquely determined by their scattering data, that is, by the right reflection
coefficient $R_+(\lam,t)$, $\lam\in\sig(H_q)$, and the eigenvalues $\rho_j\in\R\backslash\sig(H_q)$, $j=1,\dots, N$,
together with the corresponding right norming constants $\gam_{+,j}(t)>0$, $j=1,\dots, N$. Here $\sig(H_q)$ denotes
the finite-band spectrum of the underlying Lax operator $H_q$.

The relation between the energy $\lam$ of the underlying Lax operator $H_q$ and the propagation speed at which
the corresponding parts of the Toda lattice travel is given by
\be
v(\lam)=\frac{n}{t},
\ee
where
\be
v(\lam) = \lim_{\eps\to 0} \frac{-\re \int_{E_0}^{(\lam+\I\eps,+)}\! \Omega_0}{ \re \int_{E_0}^{(\lam+\I\eps,+)}\! \omega_{\infty_+\,\infty_-}},
\ee
and  can be regarded as a nonlinear analog of the classical dispersion relation. Here
$\om_{\infty_+\, \infty_-}$ is an Abelian differential of the third kind defined in \eqref{ominfpm} and
$\Omega_0$ is an Abelian differential of the second kind defined in \eqref{Om0}. We will show in
Section~\ref{secSPP} that $v$ is a homeomorphism of $\R$ and we will denote its inverse by $\zeta(n/t)$.

Define the following limiting  lattice 
\be \label{limlat}
\aligned
\prod_{j=n}^{\infty} \left(\frac{a_l(j,t)}{a_q(j,t)} \right)^2 = &
\frac{\theta(\ulz(n,t))}{\theta(\ulz(n-1,t))}
\frac{\theta(\ulz(n-1,t)+ \ul{\delta}(n,t))}{\theta(\ulz(n,t)+\ul{\delta}(n,t))} \times\\
& \times \left( \prod_{\rho_k<\zeta(n/t)} \!\!\! \exp\left(-2\int_{E(\rho)}^\rho \om_{\infty_+\, \infty_-} \right) \right) \times\\
& \times \exp\left( \frac{1}{2\pi\I} \int_{C(n/t)} 
\log (1-|R|^2) \om_{\infty_+\, \infty_-}\right),\\
\sum_{j=n}^\infty \big( b_l(j,t) - b_q(j,t) \big) = &
\frac{1}{2\pi\I} \int_{C(n/t)} \log(1-|R|^2) \Omega_0 -\\
& {} - \sum_{\rho_k<\zeta(n/t)} \int_{E(\rho_k)} \Omega_0 +\\
& {} + \frac{1}{2}\frac{d}{ds}
\log\left( \frac{\theta(\ulz(n,s) + \ul{\delta}(n,t) )}{\theta(\ulz(n,s))} \right) \Big|_{s=t},\\
\delta_\ell(n/t) &= 2 \sum_{\rho_k<\zeta(n/t)} \Amap(\hat\rho_k) +
\frac{1}{2\pi\I} \int_{C(n/t)} \log(1-|R|^2) \zeta_\ell,
\endaligned
\ee
where $R=R_+(\lam,t)$ is the associated reflection coefficient,
$\zeta_\ell$ is a canonical basis of holomorphic differentials, and
$C(n/t) = \pi^{-1}(\sig(H_q)\cap (-\infty, \zeta(n/t)))$ oriented such that the upper sheet is to the left.

\begin{theorem}\label{thm:asym}
Assume \eqref{decay} and abbreviate by $c_k= v(\rho_k)$ the velocity of the $k$'th soliton defined above.
Then the asymptotics in the soliton region, $\{ (n,t) |\, \zeta(n/t) \in\R\backslash\sig(H_q)\}$,
are as follows.

Let $\eps > 0$ sufficiently small such that the intervals $[c_k-\eps,c_k+\eps]$, $1\le k \le N$, are disjoint and lie
inside $v(\R\backslash\sig(H_q))$.

If $|\frac{n}{t} - c_k|<\eps$ for some $k$, the solution is asymptotically given by a one-soliton
solution on top of the limiting lattice:
\begin{align}\nn
\prod_{j=n}^{\infty} \frac{a(j,t)}{a_l(j,t)} &= \left(
\sqrt{\frac{c_{l,\gam_k(n,t)}(\rho_k,n-1,t)}{c_{l,\gam_k(n,t)}(\rho_k,n,t)}} + O(t^{-l}) \right),\\
\sum_{j=n+1}^\infty b(j,t) - b_l(j,t)&= -\gam_k(n,t) \frac{a_l(n,t) \psi_l(\rho_k,n,t) \psi_l(\rho_k,n+1,t)}{2 c_{l,\gam_k(n,t)}(\rho_k,n,t)}
+ O(t^{-l}),
\end{align}
for any $l \geq 1$, where
\begin{equation}
c_{l,\gam}(\rho,n,t) = 1 + \gam\!\! \sum_{j=n+1}^\infty \psi_{l,+}(\rho,j,t)^2
\end{equation}
and
\be\label{eq:gamshift}
\gam_k(n,t) = \gam_k \frac{T(\rho_k^*,n,t)}{T(\rho_k,n,t)}.
\ee
Here $\psi_l(p,n,t)$ is the Baker-Akhiezer function (cf.\ Section~\ref{secAG}) corresponding to the limiting lattice defined above.

If $|\frac{n}{t} -c_k| \geq \eps$, for all $k$, the solution is asymptotically close to the limiting lattice:
\begin{align}\nn
\prod_{j=n}^{\infty} \frac{a(j,t)}{a_l(j,t)} &= 1 + O(t^{-l}),\\
\sum_{j=n+1}^\infty b(j,t) - b_l(j,t) &= O(t^{-l}),
\end{align}
for any $l \geq 1$.
\end{theorem}

In particular, we see that the solution splits into a sum of independent solitons
where the presence of the other solitons and the radiation part corresponding to the continuous spectrum
manifests itself in phase shifts given by \eqref{eq:gamshift}. While in the constant background case
this result is classical, we are not aware of proof even in the case of a pure soliton solution.
Moreover, observe that in the periodic case considered here one can have a {\em stationary soliton}
(see the discussion in Section~\ref{secSPP}).

The proof will be given at the end of Section~\ref{sec:solreg}. Furthermore,
in the remaining regions the analysis in Section~\ref{sec:solreg} also shows
that the Riemann--Hilbert problem reduces to one without
solitons. In fact, away from the soliton region, the asymptotics can be computed as in
\cite{kt3}. The only difference being, that in the final answer, the limiting lattice has to be replaced
by the one defined here and the Blaschke factors corresponding to the eigenvalues have to be
added to the partial transmission coefficient.

Finally, we note that the same proof works even if there are different spatial
asymptotics as $n\to\pm\infty$ as long as they lie in the same isospectral class
(cf.\ Remark~\ref{rem:stepiso} below).

\section{Algebro-geometric quasi-periodic finite-gap solutions}
\label{secAG}

As in \cite{kt2}, we state some facts on our background
solution $(a_q,b_q)$ which we want to choose from the class of
algebro-geometric quasi-periodic finite-gap solutions, that is the
class of stationary solutions of the Toda hierarchy, \cite{bght}.
In particular, this class contains all periodic solutions. We will
use the same notation as in \cite{tjac}, where we also refer to for proofs.
As a reference for Riemann surfaces in this context we recommend \cite{fk}.

To set the stage let $\M$ be the Riemann surface associated with the following function
\begin{equation}\label{eq:r12}
\Rg{z}, \qquad R_{2g+2}(z) = \prod_{j=0}^{2g+1} (z-E_j), \qquad
E_0 < E_1 < \cdots < E_{2g+1},
\end{equation}
$g\in \N$. $\M$ is a compact, hyperelliptic Riemann surface of genus $g$.
We will choose $\Rg{z}$ as the fixed branch
\begin{equation}
\Rg{z} = -\prod_{j=0}^{2g+1} \sqrt{z-E_j},
\end{equation}
where $\sqrt{.}$ is the standard root with branch cut along $(-\infty,0)$.

A point on $\M$ is denoted by 
$p = (z, \pm \Rg{z}) = (z, \pm)$, $z \in \C$, or $p = (\infty,\pm) = \infty_\pm$, and
the projection onto $\C \cup \{\infty\}$ by $\pi(p) = z$. 
The points $\{(E_{j}, 0), 0 \leq j \leq 2 g+1\} \subseteq \M$ are 
called branch points and the sets 
\begin{equation}
\Pi_{\pm} = \{ (z, \pm \Rg{z}) \mid z \in \C\setminus
\bigcup_{j=0}^g[E_{2j}, E_{2j+1}]\} \subset \M
\end{equation}
are called upper, lower sheet, respectively.

Let $\{a_j, b_j\}_{j=1}^g$ be loops on the surface $\M$ representing the
canonical generators of the fundamental group $\pi_1(\M)$. We require
$a_j$ to surround the points $E_{2j-1}$, $E_{2j}$ (thereby changing sheets
twice) and $b_j$ to surround $E_0$, $E_{2j-1}$ counterclockwise on the
upper sheet, with pairwise intersection indices given by
\begin{equation}
a_i \circ a_j= b_i \circ b_j = 0, \qquad a_i \circ b_j = \delta_{i,j},
\qquad 1 \leq i, j \leq g.
\end{equation}
The corresponding canonical basis $\{\zeta_j\}_{j=1}^g$ for the space of
holomorphic differentials can be constructed by
\begin{equation}
\underline{\zeta} = \sum_{j=1}^g \underline{c}(j)  
\frac{\pi^{j-1}d\pi}{R_{2g+2}^{1/2}},
\end{equation}
where the constants $\underline{c}(.)$ are given by
\[
c_j(k) = C_{jk}^{-1}, \qquad 
C_{jk} = \int_{a_k} \frac{\pi^{j-1}d\pi}{R_{2g+2}^{1/2}} =
2 \int_{E_{2k-1}}^{E_{2k}} \frac{z^{j-1}dz}{\Rg{z}} \in
\R.
\]
The differentials fulfill
\begin{equation} \label{deftau}
\int_{a_j} \zeta_k = \delta_{j,k}, \qquad \int_{b_j} \zeta_k = \tau_{j,k}, 
\qquad \tau_{j,k} = \tau_{k, j}, \qquad 1 \leq j, k \leq g.
\end{equation}

Now pick $g$ numbers (the Dirichlet eigenvalues)
\be
(\hat{\mu}_j)_{j=1}^g = (\mu_j, \sigma_j)_{j=1}^g
\ee
whose projections lie in the spectral gaps, that is, $\mu_j\in[E_{2j-1},E_{2j}]$.
Associated with these numbers is the divisor $\dimuz$ which
is one at the points $\hat{\mu}_j$  and zero else. Using this divisor we
introduce
\begin{align} \nn
\ulz(p,n,t) &= \Amap(p) - \amap(\dimuz) - n\ul{A}_{\infty_-}(\infty_+)
+ t\ul{U}_0 - \vrc \in \C^g, \\
\ulz(n,t) &= \ulz(\infty_+,n,t),
\end{align}
where $\vrc$ is the vector of Riemann constants
\begin{equation}
\Xi_{E_0,j} = \frac{1- \sum_{k=1}^g \tau_{j,k}}{2},
\end{equation}
$\ul{U}_0$ are the $b$-periods of the Abelian differential $\Omega_0$ defined below,
and $\Amap$ ($\amap$) is Abel's map (for divisors). The hat indicates that we
regard it as a (single-valued) map from $\hat{\M}$ (the fundamental polygon
associated with $\M$ by cutting along the $a$ and $b$ cycles) to $\C^g$.
We recall that the function $\theta(\ulz(p,n,t))$ has precisely $g$ zeros
$\hmu_j(n,t)$ (with $\hmu_j(0,0)=\hmu_j$), where 
\be
\theta(\ulz) = \sum_{\ul{m} \in \Z^g} \exp 2 \pi \I \left( \spr{\ul{m}}{\ulz} +
\frac{\spr{\ul{m}}{\ul{\tau} \, \ul{m}}}{2}\right) ,\qquad \ulz \in \C^g,
\ee
is the Riemann theta function associated with $\M$.

Then our background solution is given by
\begin{align} \nn
a_q(n,t)^2 &= \ti{a}^2 \frac{\theta(\ulz(n+1,t)) \theta(\ulz(n-1,t))}{\theta(
\ulz(n,t))^2},\\ \label{imfab}
b_q(n,t) &= \tilde{b} + \frac{1}{2}
\frac{d}{dt} \log\Big(\frac{\theta(\ulz(n,t)) }{\theta(\ulz(n-1,t))}\Big).
\end{align}
The constants $\ti{a}$, $\tilde{b}$ depend only on the Riemann surface
(see \cite[Section~9.2]{tjac}).

Introduce the time dependent Baker-Akhiezer function
\begin{align}
\psi_q(p,n,t) &= C(n,0,t) \frac{\theta (\ulz(p,n,t))}{\theta(\ulz (p,0,0))}
\exp \Big( n \int_{E_0}^p \omega_{\infty_+\,\infty_-} + t\int_{E_0}^p \Omega_0
\Big),
\end{align}
where $C(n,0,t)$ is real-valued,
\begin{equation}
C(n,0,t)^2 = \frac{ \theta(\ulz(0,0)) \theta(\ulz(-1,0))}
{\theta (\ulz (n,t))\theta (\ulz (n-1,t))},
\end{equation}
and the sign has to be chosen in accordance with $a_q(n,t)$.
Here
\be\label{ominfpm}
\omega_{\infty_+\, \infty_-}= \frac{\prod_{j=1}^g (\pi -\lambda_j) }{R_{2g+2}^{1/2}}d\pi
\ee
is the Abelian differential of the third kind with poles at $\infty_+$ and $\infty_-$ and
\be\label{Om0}
\Omega_0 = \frac{\prod_{j=0}^g (\pi - \ti\lambda_j) }{R_{2g+2}^{1/2}}d\pi,
\qquad \sum_{j=0}^g \ti\lambda_j = \frac{1}{2} \sum_{j=0}^{2g+1} E_j,
\ee
is the Abelian differential of the second kind with second order poles at
$\infty_+$ and $\infty_-$ (see \cite[Sects.~13.1, 13.2]{tjac}).
All Abelian differentials are normalized to have vanishing $a_j$ periods.

We will also need the Blaschke factor
\be
B(p,\rho)= \exp \Big( g(p,\rho) \Big) = \exp\Big(\int_{E_0}^p \om_{\rho\, \rho^*}\Big) =
\exp\Big(\int_{E(\rho)}^\rho \om_{p\, p^*}\Big), \quad \pi(\rho)\in\R,
\ee
where $E(\rho)$ is $E_0$ if $\rho<E_0$, either $E_{2j-1}$ or $E_{2j}$ if
$\rho\in(E_{2j-1},E_{2j})$, $1\le j \le g$, and $E_{2g+1}$ if $\rho>E_{2g+1}$.
It is a multivalued function with a simple zero at $\rho$ and simple pole at $\rho^*$
satisfying $|B(p,\rho)|=1$, $p\in\partial\Pi_+$. It is real-valued for $\pi(p)\in(-\infty,E_0)$ and
satisfies
\be\label{eq:propblaschke}
B(E_0,\rho)=1 \quad\mbox{and}\quad
B(p^*,\rho) = B(p,\rho^*) = B(p,\rho)^{-1}
\ee
(see e.g., \cite{tag}).

The Baker-Akhiezer function is a meromorphic function on $\M\setminus\{\infty_\pm\}$
with an essential singularity at $\infty_\pm$. The two branches are denoted by
\begin{equation}
\psi_{q,\pm}(z,n,t) = \psi_q(p,n,t), \qquad p=(z,\pm)
\end{equation}
and it satisfies
\begin{align}\nn
H_q(t) \psi_q(p,n,t) &= \pi(p) \psi_q(p,n,t),\\
\frac{d}{dt} \psi_q(p,n,t) &= P_{q,2}(t) \psi_q(p,n,t),
\end{align}
where 
\begin{align}\nn
H_q(t) \psi(n) &= a_q(n,t) \psi(n+1) + a_q(n-1,t) \psi(n-1) + b_q(n) \psi(n),\\
P_{q,2}(t) \psi(n) &= a_q(n,t) \psi(n+1) - a_q(n-1,t) \psi(n-1),
\end{align}
are the operators from the Lax pair,
\be
\frac{d}{dt} H_q(t) = H_q(t) P_{q,2}(t) - P_{q,2}(t) H_q(t),
\ee
for the Toda lattice.

It is well known that the spectrum of $H_q(t)$ is time independent and
consists of $g+1$ bands
\begin{equation}
\sig(H_q) = \bigcup_{j=0}^g [E_{2j},E_{2j+1}].
\end{equation}
For further information and proofs we refer to \cite[Chap.~9 and Sect.~13.2]{tjac}.

\section{The Inverse scattering transform and the Riemann--Hilbert problem}
\label{secISTRH}

In this section our notation and results are  taken from \cite{emt} and \cite{emt2}.
Let $\psi_{q,\pm}(z,n,t)$ be the branches of the  Baker-Akhiezer function defined
in the previous section. Let $\psi_\pm(z,n,t)$ be the Jost functions for the perturbed
problem defined by
\be
\lim_{n \to \pm \infty} 
w(z)^{\mp  n} ( \psi_{\pm}(z,n,t) - \psi_{q, \pm}(z,n,t))
=0,
\ee
where $w(z)$ is the quasimomentum map
\be
w(z)= \exp (\int^p_{E_0} \omega_{\infty_+\, \infty_-}), \quad p=(z,+).
\ee
The asymptotics of the  two projections of the Jost function are
\begin{align} \nn
\psi_\pm(z,n,t) =&  \frac{z^{\mp n} \Big(\prod_{j=0}^{n-1} a_q(j,t)\Big)^{\pm 1}}{A_\pm(n,t)}  \times\\
& \times
\Big(1 + \Big(B_\pm(n,t) \pm \sum_{j=1}^{n} b_q(j- {\scriptstyle{0 \atop 1}},t) \Big)\frac{1}{z}
+ O(\frac{1}{z^2}) \Big),
\end{align}
as $z \to \infty$, where
\be \label{defABpm}
\aligned
A_+(n,t) &= \prod_{j=n}^{\infty} \frac{a(j,t)}{a_q(j,t)}, \quad
B_+(n,t)= \sum_{j=n+1}^\infty (b_q(j,t)-b(j,t)), \\
A_-(n,t) &= \!\!\prod_{j=- \infty}^{n-1}\! \frac{a(j,t)}{a_q(j,t)}, \quad
B_-(n,t) = \sum_{j=-\infty}^{n-1} (b_q(j,t) - b(j,t)).
\endaligned
\ee

One has the scattering relations
\be \label{relscat}
T(z) \psi_\mp(z,n,t) =  \ol{\psi_\pm(z,n,t)} +
R_\pm(z) \psi_\pm(z,n,t),  \qquad z \in\sigma(H_q),
\ee
where $T(z)$, $R_\pm(z)$ are the transmission respectively reflection coefficients.
Here $\psi_\pm(z,n,t)$ is defined such that 
$\psi_\pm(z,n,t)= \lim_{\eps\downarrow 0}\psi_\pm(z + \I\eps,n,t)$,
$z\in\sigma(H_q)$. If we take the limit from the other side we
have $\ol{\psi_\pm(z,n,t)}= \lim_{\eps\downarrow 0}\psi_\pm(z - \I\eps,n,t)$.

The transmission and reflection coefficients have the following well-known properties
\cite{emt}:

\begin{lemma}
The transmission coefficient $T(z)$ has a meromorphic extension to
$\C\backslash\sig(H_q)$ with simple poles at the eigenvalues $\rho_j$.
The residues of $T(z)$ are given by
\be\label{eq:resT}
\res_{\rho_j} T(z) = - \frac{R^{1/2}_{2g+2}(\rho_j)}{\prod_{k=1}^g (\rho_j-\mu_k)}
\frac{\gam_{\pm,j}}{c_j^{\pm 1}},
\ee
where
\be
\gam_{\pm,j}^{-1} = \sum_{n\in\Z} |\psi_\pm(\rho_j,n,t)|^2
\ee
and $\psi_- (\rho_j,n,t) = c_j \psi_+(\rho_j,n,t)$.

Moreover,
\be \label{reltrpm} 
T(z) \ol{R_+(z)} + \ol{T(z)} R_-(z)=0, \qquad |T(z)|^2 + |R_\pm(z)|^2=1.
\ee
\end{lemma}

In particular one reflection coefficient, say $R(z)=R_+(z)$, and one set of
norming constants, say $\gam_j= \gam_{+,j}$, suffices.

We will define a Riemann--Hilbert problem on the Riemann
surface $\M$ as follows:
\be\label{defm}
m(p,n,t)= \left\{\begin{array}{c@{\quad}l}
\begin{pmatrix} T(z) \frac{\psi_-(z,n,t)}{\psi_{q,-}(z,n,t)}  & \frac{\psi_+(z,n,t)}{\psi_{q,+}(z,n,t)} \end{pmatrix},
& p=(z,+)\\
\begin{pmatrix} \frac{\psi_+(z,n,t)}{\psi_{q,+}(z,n,t)} & T(z) \frac{\psi_-(z,n,t)}{\psi_{q,-}(z,n,t)} \end{pmatrix}, 
& p=(z,-)
\end{array}\right..
\ee
We are interested in the jump condition of $m(p,n,t)$ on $\Sigma$,
the boundary of $\Pi_\pm$ (oriented counterclockwise when viewed from top sheet $\Pi_+$).
It consists of two copies $\Sigma_\pm$ of $\sigma(H_q)$ which correspond to
non-tangential limits from $p=(z,+)$ with $\pm\im(z)>0$, respectively to non-tangential
limits from $p=(z,-)$ with $\mp\im(z)>0$.

To formulate our jump condition we use the following convention:
When representing functions on $\Sigma$, the lower subscript denotes
the non-tangential limit from $\Pi_+$ or $\Pi_-$, respectively,
\be
m_\pm(p_0) = \lim_{ \Pi_\pm \ni p\to p_0} m(p), \qquad p_0\in\Sigma.
\ee
Using the notation above implicitly assumes that these limits exist in the sense that
$m(p)$ extends to a continuous function on the boundary away from
the band edges.

Moreover, we will also use symmetries with respect to the
the sheet exchange map
\be
p^*= \begin{cases}
(z,\mp) & \text{ for } p=(z,\pm),\\
\infty_\mp & \text{ for } p=\infty_\pm,
\end{cases}
\ee
and complex conjugation
\be
\ol{p} = \begin{cases}
(\ol{z},\pm) & \text{ for } p=(z,\pm)\not\in \Sigma,\\
(z,\mp) & \text{ for } p=(z,\pm)\in \Sigma,\\
\infty_\pm & \text{ for } p=\infty_\pm.
\end{cases}
\ee
In particular, we have $\ol{p}=p^*$ for $p\in\Sigma$.

Note that we have $\ti{m}_\pm(p)=m_\mp(p^*)$ for $\ti{m}(p)= m(p^*)$
(since $*$ reverses the orientation of $\Sigma$) and $\ti{m}_\pm(p)= \ol{m_\pm(p^*)}$ for
$\ti{m}(p)=\ol{m(\ol{p})}$.

With this notation, using \eqref{relscat} and \eqref{reltrpm}, we obtain
\begin{align} \nn
m_+(p,n,t) &= m_-(p,n,t) J(p,n,t)\\ \label{rhpm2.1}
J(p,n,t) &=
\begin{pmatrix}
1 -|R(p)|^2  &- \ol{R(p) \Theta(p,n,t)}  \E^{-t \phi(p)}\\
R(p)  \Theta(p,n,t) \E^{t \phi(p)} & 1
\end{pmatrix},
\end{align}
where
\[
\Theta(p,n,t) = \frac{\theta(\ulz(p,n,t))}{\theta(\ulz(p,0,0))}
\frac{\theta(\ulz(p^*,0,0))}{\theta(\ulz(p^*,n,t))}
\]
and
\be\label{defsp}
\phi(p,\frac{n}{t}) =
2 \int_{E_0}^p \Omega_0 + 2 \frac{n}{t} \int_{E_0}^p \omega_{\infty_+\,\infty_-}
\in \I \R
\ee
for $p \in\Sigma$. Note
\[
\frac{\psi_q(p,n,t)}{\psi_q(p^*,n,t)} = \Theta(p,n,t) \E^{t\phi(p)}.
\]
Here we have extend our definition of $T$ to $\Sigma$ such that
it is equal to $T(z)$ on $\Sigma_+$ and equal to $\ol{T(z)}$ on $\Sigma_-$.
Similarly for $R(z)$. In particular, the condition on $\Sigma_+$ is just the
complex conjugate of the one on $\Sigma_-$ since we have $R(p^*)= \ol{R(p)}$
and $m_\pm(p^*,n,t)= \ol{m_\pm(p,n,t)}$ for $p\in\Sigma$.

Furthermore,
\be \label{m2infp}
m(p,n,t) = \begin{pmatrix}
A_+(n,t) (1 - 2 B_+(n-1,t) \frac{1}{z}) &
\frac{1}{A_+(n,t)}(1 + 2 B_+(n,t) \frac{1}{z} )
\end{pmatrix} + O(\frac{1}{z^2}),
\ee
for $p=(z,+)\to\infty_+$, with $A_\pm(n,t)$ and $B_\pm(n,t)$ are defined in \eqref{defABpm}.
The formula near $\infty_-$ follows by flipping the columns. Here we have used
\be
T(z) = A_-(n,t) A_+(n,t) \big( 1 - \frac{B_+(n,t) + b(n,t) + B_-(n,t)}{z} + O(\frac{1}{z^2}\big).
\ee
Using the properties of $\psi_\pm(z,n,t)$ and $\psi_{q,\pm}(z,n,t)$ one checks that its divisor satisfies
\be
(m_1) \ge -\dimus{n,t} - \dirho, \qquad (m_2) \ge -\dimu{n,t} - \dirhos,
\ee
where
\be
\dirho= \sum_j \di_{\rho_j}, \qquad \dirhos= \sum_j \di_{\rho_j^*}.
\ee
Here $(f)$ denotes the divisor of $f$.

\begin{theorem}[Vector Riemann--Hilbert problem]\label{thm:vecrhp}
Let $\mathcal{S}_+(H(0))=\{ R(\lam),\; \lam\in\sig(H_q); \: (\rho_j, \gam_j), \: 1\le j \le N \}$
the right scattering data of the operator $H(0)$. Then $m(z)=m(z,n,t)$ defined in \eqref{defm}
is meromorphic away from $\Sigma$ and satisfies:
\begin{enumerate}
\item The jump condition
\be \label{eq:jumpcond}
m_+(p)=m_-(p) J(p), \qquad
J(z)= \begin{pmatrix}
1 -|R(p)|^2  &- \ol{R(p) \Theta(p,n,t)}  \E^{-t \phi(p)}\\
R(p)  \Theta(p,n,t) \E^{t \phi(p)} & 1
\end{pmatrix},
\ee
for $p\in\Sigma$,
\item
the divisor
\be
(m_1) \ge -\dimus{n,t} - \dirho, \qquad (m_2) \ge -\dimu{n,t} - \dirhos
\ee
and pole conditions
\be\label{eq:polecond}
\aligned
& \Big( m_1(p) + \frac{R^{1/2}_{2g+2}(\rho_j)}{\prod_{k=1}^g (\rho_j-\mu_k)}
\frac{\gam_j}{\pi(p)-\rho_j} \frac{\psi_q(p,n,t)}{\psi_q(p^*,n,t)} m_2(p) \Big) \ge - \dimus{n,t},
\mbox{ near $\rho_j$},\\
& \Big( \frac{R^{1/2}_{2g+2}(\rho_j)}{\prod_{k=1}^g (\rho_j-\mu_k)}
\frac{\gam_j}{\pi(p)-\rho_j} \frac{\psi_q(p,n,t)}{\psi_q(p^*,n,t)}  m_1(p) + m_2(p) \Big) \ge - \dimu{n,t},
\mbox{ near $\rho_j^*$},
\endaligned
\ee
\item
the symmetry condition
\be \label{eq:symcond}
m(p^*) = m(p) \sigI 
\ee
\item
and the normalization
\be\label{eq:normcond}
m_1(\infty_+) \cdot m_2(\infty_+) = 1\quad m_1(\infty_+) > 0.
\ee
\end{enumerate}
\end{theorem}

\begin{proof}
We already derived the jump condition \eqref{eq:jumpcond}. The pole conditions follow since $T(p)$
is meromorphic in $\mathbb{M}\backslash\Sigma$ with simple poles at $\rho_j$ and residues given
by \eqref{eq:resT}. The symmetry condition holds by construction and the normalization \eqref{eq:normcond}
is immediate from \eqref{m2infp}.
\end{proof}

\begin{remark}\label{rem:stepiso}
We note that the same proof works even if there are different spatial
asymptotics as $n\to\pm\infty$ as long as they lie in the same isospectral class. In fact,
following \cite{emt3}, we now have two different background operators $H_q^\pm$
and we will denote the corresponding Baker-Akhiezer functions by $\psi_q^\pm(p,n,t)$.
Using
\be
m(p,n,t)= \left\{\begin{array}{c@{\quad}l}
\begin{pmatrix} T_+(z) \frac{\psi_-(z,n,t)}{\psi^-_{q,-}(z,n,t)}  & \frac{\psi_+(z,n,t)}{\psi^+_{q,+}(z,n,t)} \end{pmatrix},
& p=(z,+)\\
\begin{pmatrix} \frac{\psi_+(z,n,t)}{\psi^+_{q,+}(z,n,t)} & T_+(z) \frac{\psi_-(z,n,t)}{\psi^-_{q,-}(z,n,t)} \end{pmatrix}, 
& p=(z,-)
\end{array}\right.
\ee
in place of \eqref{defm} one easily checks that Theorem~\ref{thm:vecrhp} still holds in this case. The only difference
now is that the (right) scattering data $\rho_j$ and $R_+(z)$ will not satisfy the algebraic constraints given in \cite{tag}.
\end{remark}

For our further analysis it will be convenient to rewrite the pole condition as a jump condition.
Choose $\eps$ so small that the discs $|\pi(p)-\rho_j|<\eps$ are inside the upper sheet and
do not intersect. Then redefine $m$ in a neighborhood of $\rho_j$ respectively $\rho_j^*$ according to
\be\label{eq:redefm}
m(p) = \begin{cases}
m(p) \begin{pmatrix} 1 & 0 \\ \frac{R^{1/2}_{2g+2}(\rho_j)}{\prod_{k=1}^g (\rho_j-\mu_k)}
\frac{\gam_j}{\pi(p)-\rho_j} \frac{\psi_q(p,n,t)}{\psi_q(p^*,n,t)} & 1 \end{pmatrix},  &
\begin{smallmatrix}|\pi(p)-\rho_j|<\eps\\ p\in\Pi_+\end{smallmatrix},\\
m(p) \begin{pmatrix} 1 & \frac{R^{1/2}_{2g+2}(\rho_j)}{\prod_{k=1}^g (\rho_j-\mu_k)}
\frac{\gam_j}{\pi(p)-\rho_j} \frac{\psi_q(p,n,t)}{\psi_q(p^*,n,t)} \\
0 & 1 \end{pmatrix},  &
\begin{smallmatrix}|\pi(p)-\rho_j|<\eps\\ p\in\Pi_-\end{smallmatrix},\\
m(p), & \text{else}.\end{cases}
\ee
Then a straightforward calculation shows

\begin{lemma}\label{lem:pctoj}
Suppose $m(p)$ is redefined as in \eqref{eq:redefm}. Then $m(z)$ is meromorphic away from
$\Sigma$ and satisfies \eqref{eq:jumpcond}, \eqref{eq:symcond}, \eqref{eq:normcond},
the divisor condition change according to
\be
(m_1) \ge -\dimus{n,t}, \qquad (m_2) \ge -\dimu{n,t}
\ee
and the pole conditions are replaced by the jump conditions
\be \label{eq:jumpcond2}
\aligned
m_+(p) &= m_-(p) \begin{pmatrix} 1 & 0 \\
\frac{R^{1/2}_{2g+2}(\rho_j)}{\prod_{k=1}^g (\rho_j-\mu_k)}
\frac{\gam_j}{\pi(p)-\rho_j} \frac{\psi_q(p,n,t)}{\psi_q(p^*,n,t)} & 1 \end{pmatrix},\quad p\in\Sigma_\eps(\rho_j),\\
m_+(p) &= m_-(p) \begin{pmatrix} 1 & -\frac{R^{1/2}_{2g+2}(\rho_j)}{\prod_{k=1}^g (\rho_j-\mu_k)}
\frac{\gam_j}{\pi(p)-\rho_j} \frac{\psi_q(p,n,t)}{\psi_q(p^*,n,t)} \\
0 & 1 \end{pmatrix},\quad
p\in\Sigma_\eps(\rho_j^*),
\endaligned
\ee
where 
\be
\Sigma_\eps(p) = \{ q \in \Pi_\pm\,:\, |\pi(q)-z| =\eps \}, \qquad p=(z,\pm),
\ee
is a small circle around $p$ on the same sheet as $p$. It is oriented counterclockwise on the upper sheet and
clockwise on the lower sheet.
\end{lemma}

Next we turn to uniqueness of the solution of this vector Riemann--Hilbert problem.
This will also explain the
reason for our symmetry condition. We begin by observing that if
there are $g+1$ points $p_j\in\M$, such that $m(p_j)=\rN$, then by Riemann-Roch (\cite[Thm.~A.2]{tjac}),
\be
r(-\sum_{j=1}^{g+1} \di_{p_j}) = \deg(\sum_{j=1}^{g+1} \di_{p_j}) + 1 - g + i(\sum_{j=1}^{g+1} \di_{p_j}) \ge 2,
\ee
there are at least two linearly independent functions $h(p)$ with $(h) \ge - \sum_j \di_{p_j}$.
In particular, there is a non-constant function $h(p)$ and $n(p)= h(p) m(p)$
satisfies the same jump and pole conditions as $m(p)$. However, it will in general
violate the symmetry condition! In fact, the symmetry condition \eqref{eq:symcond} requires
$h(p)=h(p^*)$. Hence $h(p)=\ti{h}(\pi(p))$ is the lift of a meromorphic function $\ti{h}(z)$ on $\C$.
In particular, if $p_j$ is a branch point, it will have an even order pole. So if we choose the points $p_j$
to be different branch points, this yields a contradiction since at least one of the points $p_j$ must
be a non-removable pole. Thus, without the symmetry condition,
the solution of our vector Riemann--Hilbert problem will not be unique in such a situation. Moreover, 
such a situation can be indeed created starting with (e.g.) a one soliton solution and using the
Dirichlet commutation method (\cite{tsdj}, \cite[Sect.~11.8]{tjac}) to place all $g+1$ Dirichlet
eigenvalues at the band edges (which equal the branch points of our Riemann surface). 

\begin{lemma}[One soliton solution]\label{lem:singlesoliton}
Suppose there is only one eigenvalue and a vanishing reflection coefficient, that is,
$\mathcal{S}_+(H(t))=\{ R(p)\equiv 0,\; p\in\Sigma; \: (\rho, \gam) \}$.
Let
\begin{equation}
c_{q,\gam}(\rho,n,t) = 1 + \gam\!\! \sum_{j=n+1}^\infty \psi_q(\rho,j,t)^2
= 1 + \gam W_{q,(n,t)}(\psi_q'(\rho,.,t),\psi_q(\rho,.,t))
\end{equation}
and
\begin{equation}
\psi_{q,\gam}(p,n,t) = \frac{c_{q,\gam}(\rho,n,t) \psi_q(p,n,t) +
\frac{\gam}{z-\rho} \psi_q(\rho,n,t)
W_{q,(n,t)}(\psi_q(\rho,.,t),\psi_q(p,.,t))}{\sqrt{c_{q,\gam}(\rho,n-1,t)
c_{q,\gam}(\rho,n,t)}},
\end{equation}
$p=(z,\pm)$.
Here $W_{q,(n,t)}(f,g)=a_q(n,t)(f(n)g(n+1)-f(n+1)g(n))$ is the usual Wronski determinant
and the prime denotes a derivate with respect to $\rho$.

Then the unique solution of the Riemann--Hilbert problem \eqref{eq:jumpcond}--\eqref{eq:normcond}
is given by
\be\label{eq:oss}
m_0(p) = \begin{pmatrix} f(p^*,n,t) & f(p,n,t) \end{pmatrix}, \qquad
\nn f(p,n,t) = \frac{\psi_{q,\gam}(p,n,t)}{\psi_q(p,n,t)}.
\ee
In particular,
\be
A_+(n,t) = \sqrt{\frac{c_{q,\gam}(\rho,n-1,t)}{c_{q,\gam}(\rho,n,t)}}, \quad
B_+(n,t) = -\gam \frac{a_q(n,t) \psi_q(\rho,n,t) \psi_q(\rho,n+1,t)}{2 c_{q,\gam}(\rho,n,t)}.
\ee
\end{lemma}

\begin{proof}
This follows directly from \eqref{defm} after inserting the formulas for the Jost functions obtained
from (e.g.) the double commutation method found in \cite{gtjc} or \cite[Sect.~14.5]{tjac}
(cf.\ also \cite{emt4}).

Alternatively, it can also be easily checked directly, that \eqref{eq:oss}
solves \eqref{eq:jumpcond}--\eqref{eq:normcond}. In fact, except for the pole conditions
everything is straightforward. The pole conditions follow from
\begin{align*}
\lim_{p\to \rho} (z-\rho) f(p^*) &= \frac{\gam(\rho,n,t)}{\sqrt{c_{q,\gam}(\rho,n-1,t)
c_{q,\gam}(\rho,n,t)}}
W_n(\psi_q(\rho,.,t),\psi_q(\rho^*,.,t))\\
&= -\frac{\gam(\rho,n,t)}{\sqrt{c_{q,\gam}(\rho,n-1,t)
c_{q,\gam}(\rho,n,t)}} \frac{R^{1/2}_{2g+2}(\rho)}{\prod_{k=1}^g (\rho-\mu_k)},
\end{align*}
where
\[
\gam(p,n,t)= \gam \frac{\psi_q(p,n,t)}{\psi_q(p^*,n,t)} = \gam \Theta(p,n,t) \E^{t \phi(p)},
\]
and
\begin{align*}
\lim_{p\to \rho} f(p) &= \frac{c_{q,\gam}(\rho,n,t) + \gam
W_n(\psi_q(\rho,.,t),\psi_q'(\rho,.,t))}{\sqrt{c_{q,\gam}(\rho,n-1,t) c_{q,\gam}(\rho,n,t)}}\\
&= \frac{1}{\sqrt{c_{q,\gam}(\rho,n-1,t) c_{q,\gam}(\rho,n,t)}}.
\end{align*}
The formulas for $A_+(n,t)$ and $B_+(n,t)$ follow after expanding around $p=\infty_+$ and
comparing with \eqref{m2infp}.

To see uniqueness, let $\ti{m}_0(p)$ be a second solution which must be of the form
$\ti{m}_0(p) = \begin{pmatrix} \ti{f}(p^*) & \ti{f}(p) \end{pmatrix}$ by the symmetry condition.
Since the divisor $\dimus{n,t}+\dirho$ is nonspecial (\cite[Lem.~A.20]{tjac}), that is,
its index of speciality vanishes, $i(\dimus{n,t}+\dirho)=0$. Thus
the Riemann--Roch theorem (\cite[Thm.~A.2]{tjac}),
\[
r(-\dimus{n,t}-\dirho) = \deg(\dimus{n,t}+\dirho) + 1 - g + i(\dimus{n,t}+\dirho) = 2,
\]
implies that there are two linearly independent
functions satisfying $(\ti{f}) \ge -\dimus{n,t} - \dirho$. One is $f(p)$ and the other one
is the constant function. This implies $\ti{f}(p) = \alpha f(p) + \beta$ for some $\alpha,\beta\in \C$.
But the pole condition implies $\beta=0$ and the normalization condition implies $\alpha=1$.
\end{proof}

Note 
\begin{align}  \nn
a_{q,\gam}(n,t) &= a_q(n,t) \frac{\sqrt{c_{q,\gam}(\rho,n-1,t)
c_{q,\gam}(\rho,n+1,t)}}{c_{q,\gam}(\rho,n,t)},\\ \label{dc}
b_{q,\gam}(n,t) &= b_q(n,t) + \gam \partial^* 
\frac{a_q(n,t) \psi_q(\rho,n,t) \psi_q(\rho,n+1,t)}{c_{q,\gam}(\rho,n,t)},
\end{align}
where $\partial^*u(n)=u(n-1)-u(n)$, and that $a_{q,\gam}(n,t)$, $b_{q,\gam}(n,t)$ is centered at
\be
2 \alpha(\rho)(n - v(\rho) t) + \ln(\gam) = 0,
\ee
where
\be\label{solvel}
\alpha(\rho)= \re \int_{E_0}^{(\rho,+)}\!\! \omega_{\infty_+\,\infty_-}, \qquad
v(\rho)=  -\frac{1}{\alpha(\rho)} \re \int_{E_0}^{(\rho,+)}\!\! \Omega_0.
\ee
Note, however, that $H_{q,\gam}$ looks asymptotically like $H_q$ as $n\to+\infty$ but not
as $n\to-\infty$ (cf.\ \cite{emt4}).

Since $f$ has $g+1$ poles, there are also $g+1$ zeros which are given by the
Dirichlet eigenvalues of $H_{q,\gam}(t)$. Hence there is precisely one in the closure of each
interior gap of the spectrum $\sig(H_{q,\gam})=\sig(H_q) \cup \{\rho\}$.
Moreover, observe $f(p_1)=f(p_1^*)=0$ if and only if
$W_n(\psi_{q,+}(\rho),\psi_{q,+}(z_1))=W_n(\psi_{q,+}(\rho),\psi_{q,-}(z_1))=0$.
Hence we also must have $W_n(\psi_{q,+}(z_1),\psi_{q,-}(z_1))=0$, that is, $z_1\in\{ E_j \}$.
Furthermore, even in the general case $m(p_1)=\rN$ can only occur at $z_1  \in \{ E_j \}$ as the
following lemma shows.

\begin{lemma}\label{lem:resonant}
For $m(p)$ defined in \eqref{defm} set
\be \label{defhm}
\hat{m}(p)= m(p) \begin{pmatrix} \theta(\ulz(p^*,n,t)) & 0\\ 0 & \theta(\ulz(p,n,t)) \end{pmatrix}.
\ee
If $\hat{m}(p_1) = \rN$ for $m$ defined as in \eqref{defm}, then $z_1  \in \{ E_j \}$. Moreover,
the zero of at least one component is simple, in the sense that $\hat{m}_k(p)^{-1}=O((z-z_1)^{-1/2})$,
in this case.
\end{lemma}

Note that the theta functions in \eqref{defhm} are used to cancel the poles at $\uhmu$ respectively
$\uhmu^*$ such that their contribution is separated from the rest. The resulting function $\hat{m}(p)$
is multivalued (since the theta functions are), but this is of no relevance for our purpose here.

\begin{proof}
We will drop the dependence on $t$ for notational simplicity.
By \eqref{defm} the condition $\hat{m}(p_1) = \rN$ implies that the Jost solutions $\psi_-(z_1,n)$ and
$\psi_+(z_1,n)$ are linearly dependent. This can only happen, at a band edge,
$z_1 \in\{ E_j \}$, or at an eigenvalue $z_1=\rho_j$. 

We begin with the case $z_1=\rho_j$. In this case the derivative of the Wronskian
$W(z)=a(n)(\psi_+(z,n)\psi_-(z,n+1)-\psi_+(z,n+1)\psi_-(z,n))$ does not vanish
$\frac{d}{dz} W(z) |_{z=z_1} \ne 0$ (\cite[(6.11)]{emt}). Moreover,
the diagonal Green's function $g(z,n)= W(z)^{-1} \psi_+(z,n) \psi_-(z,n)$ is
Herglotz and hence can have at most a simple zero at $z=\rho_j$. Hence, if
$\psi_+(\rho_j,n) = \psi_-(\rho_j,n) =0$, both can have at most a simple zero at $z=\rho_j$.
But $T(z)$ has a simple pole at $\rho_j$ and hence $T(z) \psi_-(z,n)$ cannot
vanish at $z=\rho_j$, a contradiction.

It remains to show that one zero is simple in the case $z_1\in\{ E_j \}$. To show this
we will introduce a local coordinate $\zeta=\sqrt{z-z_1}$ and show that
$\frac{d}{d\zeta} W(\zeta) |_{\zeta=0} \ne 0$ in this case as follows: 
We will consider every function of $z$ as a function of $\zeta$ and, by abuse of notation,
denote it by the same name. Moreover, we will assume $z_1\ne \mu_j$ for all $j$.
Then, note that $\psi_\pm'(\zeta)$ (where $\prime$ denotes the derivative with respect to
$\zeta$) again solves $H\psi_\pm'(0) = z_1 \psi_\pm'(0)$ if $z_1\in\{ E_j \}$. Moreover, by
$W(0)=0$ we have $\psi_+(0) = c \psi_-(0)$ for some constant $c$ (independent of $n$).
Thus we can compute
\begin{align*}
W'(0) &= W(\psi_+'(0),\psi_-(0)) + W(\psi_+(0),\psi_-'(0))\\
&= c^{-1} W(\psi_+'(0),\psi_+(0)) + c W(\psi_-(0),\psi_-'(0))
\end{align*}
by letting $n\to+\infty$ for the first and $n\to-\infty$ for the second Wronskian (in which case we can
replace $\psi_\pm(0,n)$ by $\psi_{q,\pm}(0,n)$),
which gives (cf.\ \cite[Chap.~6]{tjac})
\[
W'(0) = \frac{c+c^{-1}}{2 \prod_{j=1}^g(z_1-\mu_j)} \lim_{\zeta\to 0} \frac{R^{1/2}_{2g+2}(z_1+\zeta^2)}{\zeta} \ne 0.
\]
Hence the Wronskian has a simple zero. But if both functions had more than
simple zeros, so would the Wronskian, a contradiction. The case where $z_1= \mu_j$ for some $j$ is similar.
Just observe that in this case both $\psi_\pm$ as well as $\psi_{q,\pm}$ have singularities, which however
cancel in the quotient appearing in $\hat{m}(p)$.
\end{proof}

\section{A uniqueness result for symmetric vector Riemann--Hilbert problems}

In this section we want to investigate uniqueness for the meromorphic vector Riemann--Hilbert problem 
\begin{align}\nn
& m_+(p) = m_-(p) v(p), \qquad z\in \Sigma,\\ \nn
& (m_1) \ge -\dimuzs, \qquad (m_2) \ge -\dimuz\\ \label{eq:rhp4m}
& m(p^*) = m(p) \sigI,\\ \nn
& m(\infty_+) = \begin{pmatrix} 1 & m_2(\infty_+)\end{pmatrix}.
\end{align}
where $\Sigma$ is a nice oriented contour (see Hypothesis~\ref{hyp:sym}), symmetric with respect to
$p\mapsto p^\ast$, and $v$ is continuous satisfying
\be
v(p^*) = \sigI v(p) \sigI,\quad p\in\Sigma.
\ee
The normalization used here will be more convenient than \eqref{eq:normcond}.
In fact, \eqref{eq:normcond} will be satisfied by $m_2^{-1/2}(\infty_+) m(p)$.

Now we are ready to show that the symmetry condition in fact guarantees uniqueness.

\begin{theorem}
Suppose there exists a solution $m(p)$ of the Riemann--Hilbert problem \eqref{eq:rhp4m} 
such that for
\be\label{eq:rhp4hm}
\hat{m}(p)= m(p) \begin{pmatrix} \theta(\ulz(p^*)) & 0\\ 0 & \theta(\ulz(p)) \end{pmatrix}
\ee
the equality $\hat{m}(p)=\begin{pmatrix} 0 & 0\end{pmatrix}$ can happen at most for $p \in \{ E_j\}$ in which case 
$\limsup_{z\to E_j} \frac{\sqrt{z-E_j}}{\hat{m}_j(p)}$ is bounded from any direction for $j=1$ or $j=2$.

Then the Riemann--Hilbert problem \eqref{eq:rhp4m} with norming condition replaced by
\be\label{eq:rhp4ma}
m(\infty_+) = \begin{pmatrix} \alpha & m_2(\infty_+)\end{pmatrix}
\ee
for given $\alpha\in\C$, has a unique solution $m_\alpha(z) = \alpha\, m(z)$.
\end{theorem}

Again, the theta functions in \eqref{eq:rhp4hm} are used to cancel the poles at $\uhmuz$ and
$\uhmuz^*$ such that their contribution is separated from the rest.

\begin{proof}
Let $m_\alpha(z)$ be a solution of \eqref{eq:rhp4m} normalized according to
\eqref{eq:rhp4ma}. Then we can construct a matrix valued solution via $M=(m, m_\alpha)$ and
there are two possible cases: Either $\det M(z)$ is nonzero for some $z$ or it vanishes
identically.

We start with the first case. By the Lemma~\ref{lem:pctoj}, we can rewrite all
poles as jumps with determinant one. Hence, the determinant of this
modified Riemann--Hilbert problem has no jump. Hence it is a meromorphic function
whose divisor satisfies $(\det(M)) \ge -\dimuz - \dimuzs$. Since $\dimuz$ is a nonspecial
divisor, so is $\dimuz + \dimuzs$ and the Riemann--Roch theorem implies that there
are $g+1$ linearly independent meromorphic functions of this kind. By inspection they are
given by
\be
\det(M(p)) = \frac{P(z)}{\prod_{j=1}^g(z-\mu_j)}, \qquad z=\pi(p),
\ee
where $P(z)$ is a polynomial of degree at most $g$. But taking determinants in
\[
M(p^*) = M(p) \sigI.
\]
gives a contradiction.

It remains to investigate the case where $\det(M)\equiv 0$. In this case
we have $m_\alpha(p) = \delta(p) m(p)$ with a scalar function $\delta$. Moreover,
$\delta(p)$ must be holomorphic for $z\in\C\backslash\Sigma$ and continuous
for $z\in\Sigma$ except possibly at the points where $m(z_0) = \rN$. Since it has
no jump across $\Sigma$,
\[
\delta_+(p) m_+(p) = m_{\alpha,+}(p) = m_{\alpha,-}(p) v(p) = \delta_-(p) m_-(p) v(p)
= \delta_-(p) m_+(p),
\]
it is meromorphic with divisor
\[
(\delta) \ge - \sum_j \di_{E_j}.
\]
But the symmetry $\delta(p) = \delta(p^*)$ requires at least second order poles at the branch points (if any
at all), which shows that $\delta$ is constant. This finishes the proof.
\end{proof}

Furthermore, note that the requirements cannot be relaxed to allow (e.g.) second order
zeros in stead of simple zeros. In fact, if $m(p)$ is a solution for which both components
vanish of second order at, say, $p=E_j$, then $\ti{m}(p)=\frac{1}{z-E_j} m(p)$ is a
nontrivial symmetric solution of the vanishing problem (i.e.\ for $\alpha=0$).

By Lemma~\ref{lem:resonant} we have

\begin{corollary}\label{cor:unique}
The function $m(p,n,t)$ defined in \eqref{defm} is the only solution of the
vector Riemann--Hilbert problem \eqref{eq:jumpcond}--\eqref{eq:normcond}.
\end{corollary}

\section{The stationary phase points and the nonlinear dispersion relation}
\label{secSPP}

In this section we want to look at the relation between the energy $\lam$ of the underlying
Lax operator $H_q$ and the propagation speed at which the corresponding parts of the Toda
lattice travel, that is, the analog of the classical dispersion relation. If we set
\be\label{def:v}
v(\lam) = \lim_{\eps\to 0} \frac{-\re \int_{E_0}^{(\lam+\I\eps,+)}\! \Omega_0}{ \re \int_{E_0}^{(\lam+\I\eps,+)}\! \omega_{\infty_+\,\infty_-}},
\ee
the nonlinear dispersion relation is given by
\be
v(\lam)=\frac{n}{t}.
\ee
Recall that the Abelian differentials are given by \eqref{ominfpm} and \eqref{Om0}.

For $\rho\in\R\backslash\sig(H_q)$ we have
\be
v(\rho) = \frac{-\int_{E_0}^{(\rho,+)}\! \Omega_0}{
\int_{E_0}^{(\rho,+)}\! \omega_{\infty_+\,\infty_-}},
\ee
that is,
\be
v(\rho) = \frac{n}{t} \quad\Leftrightarrow\quad \phi(\rho,\frac{n}{t}) =0.
\ee
In other words, $v(\rho)$ is precisely the velocity of a soliton
corresponding to the eigenvalue $\rho$.

For $\lam\in\sig(H_q)$ both nominator and denominator vanish on $\sig(H_q)$. Hence
by de l'Hospital we get
\be
v(\lam) = -\frac{\prod_{j=0}^g (\lam -\ti\lambda_j)}{\prod_{j=1}^g (\lam -\lambda_j)},
\ee
that is,
\be
v(\lam) = \frac{n}{t} \quad\Leftrightarrow\quad \phi'(\lam,\frac{n}{t}) =0,
\ee
where $\phi$, defined in \eqref{defsp},  is the phase of factorization problem
\eqref{rhpm2.1}. In other words, $v(\lam) = \frac{n}{t}$ if and only if $\lam$ is a stationary
phase point.

Invoking \eqref{ominfpm} and \eqref{Om0}, we see that the stationary phase points are given by 
\be\label{eq:statph}
\prod_{j=0}^g (z -\tilde\lam_j) + \frac{n}{t} \prod_{j=1}^g (z -\lam_j) =0.
\ee
Due to the normalization of our Abelian differentials, the numbers
$\lam_j$, $1\le j \le g$, are real and different with precisely one lying in each
spectral gap, say $\lam_j$ in the $j$'th gap.
Similarly, $\ti\lam_j$, $0\le j \le g$, are real and different and
$\ti\lam_j$, $1\le j \le g$, sits in the $j$'th gap. However $\ti\lam_0$ can be
anywhere (see \cite[Sect.~13.5]{tjac}).

\begin{lemma}[\cite{kt2}]
Denote by $z_j(v)$, $0\le j \le g$, the stationary phase points, where
$v=\frac{n}{t}$. Set $\lam_0=-\infty$ and $\lam_{g+1}= \infty$, then
\be
\lam_j < z_j(v) < \lam_{j+1}
\ee
and there is always at least one stationary phase point in the $j$'th spectral gap.
Moreover, $z_j(v)$ is monotone decreasing with
\be
\lim_{v\to-\infty} z_j(v) = \lam_{j+1} \quad\text{and}\quad
\lim_{v\to\infty} z_j(v) = \lam_j.
\ee
\end{lemma}

So, depending on $n/t$ there is at most  one single stationary phase point belonging
to the union of the bands $\sigma(H_q)$, say $z_j(v)$.

We now can establish that $v(\lam)$ is monotone.

\begin{lemma}
The function $v(\rho)$ defined in \eqref{def:v} is continuous and strictly monotone decreasing.
Moreover, it is a bijection from $\R$ to $\R$.
\end{lemma}

\begin{proof}
First of all observe that $v(\lam)$ is continuous. This is obvious except at the band edges $\lam=E_j$. However,
computing $\lim_{\lam\to E_j} v(\lam)$ using again de l'Hospital establishes continuity at these points as well.

Furthermore, for large $\rho$ we have
\be
\lim_{|\rho|\to\infty} \frac{v(\rho)}{-\rho/\log(|\rho|)} =1,
\ee
which shows $\lim_{\rho\to\pm\infty} v(\rho) = \mp\infty$.

In the regions, where there is one stationary phase point $z_j(v)\in \sigma(H_q)$ we know that $z_j(v)$
is the inverse of $v(\lam)$ and monotonicity follows from the previous lemmas. In the other regions we
obtain from $v(\zeta(z))=z$ by the implicit function theorem
\be
\zeta' = -\Rg{\zeta} \frac{\re\int_{E_0}^{(\zeta,+)}\!\! \omega_{\infty_+\,\infty_-}}{\prod_{j=0}^g (\zeta -\ti\lam_j) + v \prod_{j=1}^g (\zeta -\lam_j)}
= -\Rg{\zeta} \frac{\re\int_{E_0}^{(\zeta,+)}\!\! \omega_{\infty_+\,\infty_-}}{\prod_{j=0}^g (\zeta -z_j(v))}
\ee
which shows strict monotonicity since $\re\int_{E_0}^{(\zeta,+)}\!\! \omega_{\infty_+\,\infty_-}>0$ for $\zeta\in\R\backslash\sig(H_q)$
and $z_j(v) \le \zeta(v) \le z_{j-1}(v)$ for $\zeta(v)\in(E_{2j-1},E_{2j})$ (if we set $z_{-1}=\infty$, $z_{g+1}=-\infty$, $E_{-1}=-\infty$, $E_{2g+2}=\infty$).
To see the last claim we can argue as follows: If $\zeta(v)$ were below $z_j(v)$ at some point it would decrease as $v$ decreases
whereas $z_j(v)$ increases as $v$ decreases. This contradicts the fact that both must hit at $E_{2j-1}$. Similarly we see that
$\zeta(v)$ stays below $z_j(v)$.
\end{proof}

In summary, we can define a function $\zeta(n/t)$ via
\be\label{def:zeta}
v(\zeta)=\frac{n}{t}.
\ee
In particular, different solitons travel at different speeds
and don't collide with each other or the parts corresponding to
the continuous spectrum.

Moreover, there is some $\zeta_0$ for which $v(\zeta_0)=0$ and hence there can be {\em stationary solitons}
provided $\zeta_0\not\in\sig(H_q)$. To show that this can indeed happen and to shed some further light on
the location of this point we establish the following facts:

\begin{lemma}
\begin{enumerate}
\item
$\ti{\lam}_0$ satisfies
\be\label{eq:condtilam0}
\frac{E_{2g+1}-E_0}{2} - \frac{1}{2}|\sig(H_q)| < \ti{\lam}_0 < \frac{E_{2g+1}-E_0}{2} + \frac{1}{2} |\sig(H_q)|,
\ee
where
$|\sig(H_q)| = \sum_{j=0}^g (E_{2j+1} -E_{2j})$ denotes the Lebesgue measure of $\sig(H_q)$.
\item
There exists a unique $\zeta_0$ such that $v(\zeta_0) = 0$ which lies
inside the convex hull of the spectrum $\sig(H_q)$. Furthermore, if $\zeta_0 \in\sig(H_0)$ or $\ti{\lam}_0 \in\sig(H_0)$,
then $\zeta_0 = \ti{\lam}_0$.
\end{enumerate}
\end{lemma}

\begin{proof}
(i). Observe that by \eqref{Om0}
\[
\ti{\lam}_0 = \frac{1}{2} \sum_{j=0}^{2g +1} E_j - \sum_{j=1}^g \ti{\lam_j}
=\frac{E_0+E_{2g+1}}{2} + \sum_{j=1}^g \left(\frac{E_{2j-1}+ E_{2j}}{2} -\ti{\lam}_j\right)
\]
and the claim follows using $\ti{\lam}_j \in (E_{2j-1}, E_{2j})$, $j = 1,\dots,g$.

(ii). Existence and uniqueness of $\zeta_0$ follows since $v$ is a bijection.
To check that $\zeta_0$ is in the convex hull of $\sig(H_q)$, it suffices to
check that $v(E_0) > 0$ and $v(E_{2g +1}) < 0$, this follows from
\eqref{eq:statph} using $\lam_j, \ti{\lam}_j \in (E_{2j-1}, E_{2j})$, and \eqref{eq:condtilam0}.
That $\zeta_0 = \ti{\lam}_0$ if $\zeta\in\sig(H_q)$ or $\ti{\lam}_0 \in\sig(H_0)$ follows again from \eqref{eq:statph}
and $\ti{\lam}_j \in (E_{2j-1}, E_{2j})$.
\end{proof}

This lemma implies in particular that $\zeta_0\in\sig(H_q)$ if and only
if $\ti{\lam}_0 \in\sig(H_q)$. To see that both possibilities can occur, observe that
if the band edges are symmetric with respect to $\lam\mapsto-\lam$, the same will be true for
$\lam_j$ and $\ti{\lam}_j$.

\section{The partial transmission coefficient}

Define a divisor $\dinu{n,t}$ of degree $g$ via
\be
\amap(\dinu{n,t}) = \amap(\dimu{n,t}) + \ul{\delta}(n/t),
\ee
where
\be \label{defdel}
\delta_\ell(n/t) = 2 \sum_{\rho_k<\zeta(n/t)} \Amap(\hat\rho_k) +
\frac{1}{2\pi\I} \int_{C(n/t)} \log(1-|R|^2) \zeta_\ell,
\ee
where $C(n/t) = \Sigma \cap \pi^{-1}((-\infty,\zeta(n/t))$ and $\zeta(n/t)$ is defined in \eqref{def:zeta}.

Then $\dinu{n,t}$ is nonspecial and $\pi(\hat{\nu}_j(n,t))=\nu_j(n,t)\in\R$ with precisely one
in each spectral gap (see \cite{kt2}).

Then we define the partial transmission coefficient as
\be \label{defd}
\aligned
T(p,n,t) = & \left(\frac{\theta(\ulz(n,t)+\ul{\delta}(n/t))}{\theta(\ulz(n,t))}
\frac{\theta(\ulz(n-1,t)+\ul{\delta}(n/t))}{\theta(\ulz(n-1,t))}\right)^{1/2}
\frac{\theta(\ulz(p,n,t))}{\theta(\ulz(p,n,t)+ \ul{\delta}(n/t))}\\
& \times  
\left( \prod_{\rho_k<\zeta(n/t)} \!\!\! \exp\left(-\int_{E_0}^p \om_{\rho_k\, \rho_k^*} \right) \right)
\exp\left( \frac{1}{2\pi\I} \int_{C(n/t)} \!\!\! \log (1-|R|^2) \omega_{p E_0}\right),
\endaligned
\ee
where $\ul{\delta}(n,t)$ is defined in \eqref{defdel} and $\omega_{p q}$
is the Abelian differential of the third kind with poles at $p$ and $q$. In the case where
we have the full transmission coefficient this formula was derived in \cite{tag} (see also
\cite{emt} and \cite{emt3}).

The function $T(p,n,t)$ is meromorphic in $\M\setminus C(n/t)$ with first order poles at
$\rho_k<\zeta(n/t)$, $\hat{\nu}_j(n,t)$ and first order zeros at $\hat{\mu}_j(n,t)$.

\begin{lemma}
$T(p,n,t)$ satisfies the following scalar meromorphic Riemann--Hilbert problem:
\be \label{rhpptc}
\aligned
&T_+(p,n,t) = T_-(p,n,t) (1-|R(p)|^2), \quad p \in C(n/t),\\
&(T(p,n,t))= \sum_{\rho_k<\zeta(n/t)} \di_{\rho_k^*} -\sum_{\rho_k<\zeta(n/t)} \di_{\rho_k} + \dimu{n,t}  -\dinu{n,t},\\
&T(\infty_+,n,t) T(\infty_-,n,t)= 1, \quad T(\infty_+,n,t)>0.
\endaligned
\ee
Moreover,
\begin{enumerate}
\item
\[
T(p^*,n,t) T(p,n,t) = \prod_{j=1}^g \frac{z-\mu_j}{z-\nu_j}, \qquad z=\pi(p).
\]
\item
$\ol{T(p,n,t)}=T(\ol{p},n,t)$ and in particular $T(p,n,t)$ is real-valued for $p\in\Sigma$.
\end{enumerate}
\end{lemma}

\begin{proof}
This can be shown as in \cite[Thm.~4.3]{kt2}.
\end{proof}

We will also need the expansion around $\infty_+$ given by
\be\label{Tinfp}
T(p,n,t) = T_0(n,t) \left(1 + \frac{T_1(n,t)}{z} + O(\frac{1}{z^2}) \right), \qquad p=(z,+),
\ee
where
\begin{align} \nn
T_0(n,t) = &T(\infty_+,n,t) = \left(\frac{\theta(\ulz(n,t))}{\theta(\ulz(n,t)+\ul{\delta}(n/t))}
\frac{\theta(\ulz(n-1,t)+\ul{\delta}(n/t))}{\theta(\ulz(n-1,t))}\right)^{1/2}\\
& \times  
\left( \prod_{\rho_k<\zeta(n/t)} \!\!\! \exp\left(-\int_{E(\rho)}^\rho \om_{\infty_+\, \infty_-} \right) \right)
 \exp\left( \frac{1}{4\pi\I} \int_{C(n/t)} \!\!\! \log (1-|R|^2) \omega_{\infty_+ \infty_-}\right),
\end{align}
and
\begin{align*}
T_1(n,t) =& \sum_{\rho_k<\zeta(n/t)} \int_{E(\rho_k)} \Omega_0
- \frac{1}{2\pi\I} \int_{C(n/t)} \log(1-|R|^2) \Omega_0 - \\
& {} - \frac{1}{2}\frac{d}{ds} \log\left( \frac{\theta(\ulz(n,s) + \ul{\delta}(n,t) )}{\theta(\ulz(n,s))} \right) \Big|_{s=t},
\end{align*}
where $\Omega_0$ is the Abelian differential of the second kind defined in \eqref{ominfpm}.
This follows as in \cite[Sect.~4]{tag}.

\section{Solitons and the soliton region}
\label{sec:solreg}

This section demonstrates the basic method of passing from a Riemann--Hilbert problem
involving solitons to one without. Solitons are represented in a Riemann--Hilbert problem by pole conditions,
for this reason we will further study how poles can be dealt with in this section.
We follow closely the presentation in Section~4 of \cite{krt}.

In order to remove the poles there are two cases to distinguish. If $\rho_j > \zeta(n/t)$ the jump
is exponentially close to the identity and there is nothing to do. 

Otherwise we need to use conjugation to turn the jumps into this form exponentially decaying
ones, again following Deift, Kamvissis, Kriecherbauer, and Zhou \cite{dkkz}.
It turns out that we will have to handle the poles at $\rho_j$ and $\rho_j^*$
in one step in order to preserve symmetry and in order to not add additional poles
elsewhere.

For easy reference we note the following result which can be checked by a straightforward
calculation.

\begin{lemma}[Conjugation]\label{lem:conjug}
Assume that $\widetilde{\Sigma}\subseteq\Sigma$. Let $D$ be a matrix of the form
\be
D(p) = \begin{pmatrix} d(p^*) & 0 \\ 0 & d(p) \end{pmatrix},
\ee
where $d: \mathbb{M}\backslash\widetilde{\Sigma}\to\C$ is a sectionally analytic function. Set
\be
\ti{m}(p) = m(p) D(p),
\ee
then the jump matrix transforms according to
\be
\ti{v}(p) = D_-(p)^{-1} v(p) D_+(p).
\ee
$\ti{m}(p)$ satisfies \eqref{eq:symcond} if and only if $m(p)$ does.
Furthermore, $\ti{m}(p)$ satisfies \eqref{eq:normcond}, if $m(p)$
satisfies \eqref{eq:normcond} and $d(\infty_\pm) d(\infty_\mp) = 1$.
\end{lemma}

In contradistinction to \cite{krt}, we will no longer have $\det(D(p)) = 1$,
but $\det(D(p))$ will be the lift of a rational function with $g$ zeros and poles.

\begin{lemma}\label{lem:intib}
Introduce
\be
\ti{B}(p,\rho) =
C_\rho(n,t)
\frac{\theta(\ul{z}(p,n,t))}{\theta(\ul{z}(p,n,t) + 2 \Amap(\rho))} B(p,\rho).
\ee
Then $\ti{B}(.,\rho)$ is a well defined meromorphic function, with divisor
\be
(\ti{B}(.,\rho)) = -\mathcal{D}_{\hat{\ul{\nu}}} + \mathcal{D}_{\hat{\ul{\mu}}}
- \mathcal{D}_{\rho^*} + \mathcal{D}_{\rho},
\ee
where $\nu$ is defined via
\be
\ul{\alpha}_{E_0}(\mathcal{D}_{\hat{\ul{\nu}}}) = 
\ul{\alpha}_{E_0}(\mathcal{D}_{\hat{\ul{\mu}}}) + 2 \Amap(\rho).
\ee
Furthermore, $\ti{B}(., \rho)$ has a pole at $\rho^*$, and
\be\label{eq:normtib}
\ti{B}(\infty_+,\rho) \ti{B}(\infty_-,\rho) = 1,
\ee
if
\be
C_\rho(n,t)^2 = \frac{\theta(\ul{z}(n,t) + 2 \Amap(\rho))}{\theta(\ul{z}(n,t))}
\frac{\theta(\ul{z}(n-1,t) + 2 \Amap(\rho))}{\theta(\ul{z}(n-1,t))}.
\ee
\end{lemma}

\begin{proof}
We start by checking single valuedness. The $a$-periods follow from normalization.
For the $b$ periods, we compute for $E \in (E_{2l-1}, E_{2l})$
\[
\lim_{\eps\downarrow 0} \frac{B(E + i \eps,\rho)}{B(E - i \eps,\rho)} = \exp\left( 2\pi\I (2 A_{E_0,l}(\rho))\right)
\]
using (A.21) in \cite{tjac}. Now using (A.68) in \cite{tjac} the claim follows.

The normalization condition \eqref{eq:normtib} follows by a computation using \eqref{eq:propblaschke}.
\end{proof}

Now, we can show how to conjugate the jump corresponding to one eigenvalue.

\begin{lemma}\label{lem:twopolesinc}
Assume that the Riemann--Hilbert problem for $m$ has jump conditions near $\rho$ and
$\rho^*$ given by
\be
\aligned
m_+(p)&=m_-(p)\begin{pmatrix}1& 0 \\ \frac{\gam(p)}{\pi(p) - \rho} &1\end{pmatrix}, \qquad p\in\Sigma_\eps(\rho),\\
m_+(p)&=m_-(p)\begin{pmatrix}1& -\frac{\gam(p^*)}{\pi(p) - \rho} \\ 0 &1\end{pmatrix}, \qquad p\in\Sigma_\eps(\rho^*),
\endaligned
\ee
and satisfies a divisor condition 
\be
(m_1) \ge -\dimuzs, \qquad (m_2) \ge -\dimuz.
\ee
Then this Riemann--Hilbert problem is equivalent to a Riemann--Hilbert problem for $\ti{m}$
which has jump conditions near $\rho$ and $\rho^*$ given by
\be
\aligned
\ti{m}_+(p)&= \ti{m}_-(p)\begin{pmatrix}1& \frac{\ti{B}(p,\rho^*) (\pi(p) - \rho)}{\gam(p) \ti{B}(p^*,\rho^*)} \\ 0 &1\end{pmatrix},
\qquad p\in\Sigma_\eps(\rho),\\
\ti{m}_+(p)&= \ti{m}_-(p)\begin{pmatrix}1& 0 \\ -\frac{\ti{B}(p^*,\rho^*) (\pi(p) - \rho)}{\gam(p^*) \ti{B}(p,\rho^*)} &1\end{pmatrix},
\qquad p\in\Sigma_\eps(\rho^*),
\endaligned
\ee
divisor condition
\be
(\ti{m}_1) \ge -\dinuzs, \qquad (\ti{m}_2) \ge -\dinuz,
\ee
where $\dinuz$ is defined via
\be
\ul{\alpha}_{E_0}(\mathcal{D}_{\hat{\ul{\nu}}}) = 
\ul{\alpha}_{E_0}(\mathcal{D}_{\hat{\ul{\mu}}}) + 2 \Amap(\rho),
\ee
and all remaining data conjugated (as in Lemma~\ref{lem:conjug}) by
\be
D(p) = \begin{pmatrix} \ti{B}(p^*,\rho^*) & 0 \\ 0 & \ti{B}(p,\rho^*) \end{pmatrix}.
\ee
\end{lemma}

\begin{proof}
Denote by $U$ the interior of $\Sigma_\eps(\rho)$.
To turn $\gam$ into $\gam^{-1}$, introduce $D$ by
\[
D(p) = \begin{cases}
\begin{pmatrix} 1 & \frac{\pi(p) - \rho}{\gam(p)} \\
-\frac{\gam(p)}{\pi(p) - \rho} & 0 \end{pmatrix}
\begin{pmatrix} \ti{B}(p^*,\rho^*) & 0 \\ 0 & \ti{B}(p,\rho^*) \end{pmatrix}, &  p \in U, \\
\begin{pmatrix} 0 & -\frac{\gam(p^*)}{\pi(p) - \rho} \\ \frac{\pi(p) - \rho}{\gam(p^*)} & 1 \end{pmatrix}
\begin{pmatrix} \ti{B}(p^*,\rho^*) & 0 \\ 0 & \ti{B}(p,\rho^*) \end{pmatrix}, & p^* \in U, \\ 
\begin{pmatrix} \ti{B}(p^*,\rho^*) & 0 \\ 0 & \ti{B}(p,\rho^*) \end{pmatrix}, & \text{else},
\end{cases}
\]
and note that $D(p)$ is meromorphic away from the two circles. Now set $\ti{m}(p) = m(p) D(p)$.
The claim about the divisors follows from noting, where the poles of $\ti{B}(p,\rho)$ are.
\end{proof}

This result can be applied iteratively to conjugate all eigenvalues $\rho_j < \zeta(n/t)$ as
follows: One starts with the original poles $\mu = \mu^0$ and applies the lemma
with $\rho = \rho_1$ resulting in new poles $\mu^1 = \nu$. Then one repeats
this with $\mu = \mu^1$, $\rho = \rho_2$, and so on. Finally, we will also need a
last conjugation step to factor our jump matrices into upper and lower triangular
parts as demonstrated in \cite{kt2}. Combining all steps we end up with the following
conjugation:

Abbreviate
\[
\gam_k(p,n,t)= 
\frac{R^{1/2}_{2g+2}(\rho_k)}{\prod_{l=1}^g (\rho_k-\mu_l)} \frac{\psi_q(p,n,t)}{\psi_q(p^*,n,t)} \gam_k
\]
and introduce
\[
D(p) = \begin{cases}
\begin{pmatrix} 1 & \frac{\pi(p) - \rho_k}{\gam_k(p,n,t)} \\ -\frac{\gam_k(p,n,t)}{\pi(p) - \rho_k} & 0 \end{pmatrix}
D_0(p), & \begin{smallmatrix}|\pi(p)-\rho_k|<\eps\\ p\in\Pi_+\end{smallmatrix}, \: \rho_k < \zeta(n/t),\\
\begin{pmatrix} 0 & -\frac{\gam_k(p^*,n,t)}{\pi(p) - \rho_k} \\ \frac{\pi(p) - \rho_k}{\gam_k(p^*,n,t)} & 1 \end{pmatrix}
D_0(p), & \begin{smallmatrix}|\pi(p)-\rho_k|<\eps\\ p\in\Pi_-\end{smallmatrix}, \: \rho_k < \zeta(n/t),\\ 
D_0(p), & \text{else},
\end{cases}
\]
where 
\[
D_0(p) = \begin{pmatrix} T(p^\ast,n,t) & 0 \\ 0 & T(p,n,t) \end{pmatrix}.
\]
Note that $D(p)$ is meromorphic in $\mathbb{M}\backslash\Sigma(z_0)$ and that we have
\[
D(p^\ast)= \sigI D(p) \sigI.
\]
Now we conjugate our problem using $D(p)$ and observe that, since $T(p,n,t)$ has
the same behaviour as $T(p)$ for $p$ a band edge, the new vector $\ti{m}(p)=m(p) D(p)$
is again continuous near the band edges.

Then, the divisor conditions are shifted
\be
(\ti{m}_1) \geq -\mathcal{D}_{\uhnu^*},
\quad (\ti{m}_2) \geq -\mathcal{D}_{\uhnu}.
\ee
Moreover, using Lemma~\ref{lem:conjug} and Lemma~\ref{lem:twopolesinc} the jump
corresponding $\rho_k < \zeta(n/t)$ (if any) is given by
\be
\aligned
\ti{v}(p) &= \begin{pmatrix}1& \frac{T(p,n,t) (\pi(p) - \rho_k)}{\gam_k(p,n,t) T(p^*,n,t)} \\ 0 &1\end{pmatrix},
\qquad p\in\Sigma_\eps(\rho_k),\\
\ti{v}(p) &= \begin{pmatrix}1& 0 \\ -\frac{T(p^*,n,t) (\pi(p) - \rho_k)}{\gam_k(p^*,n,t) T(p,n,t)} &1\end{pmatrix},
\qquad p\in\Sigma_\eps(\rho_k^*),
\endaligned
\ee
and corresponding $\rho_k > \zeta(n/t)$ (if any) by
\be
\aligned
\ti{v}(p) &= 
\begin{pmatrix} 1 & 0 \\
\frac{\gam_k(p,n,t) T(p^*,n,t)}{T(p,n,t) (\pi(p)-\rho_k)} & 1 \end{pmatrix},
\qquad p\in\Sigma_\eps(\rho_k),\\
\ti{v}(p) &= 
\begin{pmatrix} 1 & -\frac{\gam_k(p^*,n,t) T(p,n,t)}{T(p^*,n,t) (\pi(p)-\rho_k)} \\ 
0 & 1 \end{pmatrix},
\qquad p\in\Sigma_\eps(\rho_k^*).
\endaligned
\ee
In particular, all jumps corresponding to poles, except for possibly one if
$\rho_k=\zeta(n/t)$, are exponentially decreasing. In this case we will keep the
pole condition which now reads
\be
\aligned
& \Big( \ti{m}_1(p) + \frac{\gam_k(p,n,t) T(p^*,n,t)}{T(p,n,t) (\pi(p) - \rho_k)} \ti{m}_2(p) \Big) \ge - \dinus{n,t},
\mbox{ near $\rho_k$},\\
& \Big( \frac{\gam_k(p^*,n,t) T(p,n,t)}{T(p^*,n,t) (\pi(p) - \rho_k)}  \ti{m}_1(p) + \ti{m}_2(p) \Big) \ge - \dinu{n,t},
\mbox{ near $\rho_k^*$},
\endaligned
\ee
Furthermore, the jump along
$\Sigma$ is given by
\be
\ti{v}(p) = \begin{cases}
\ti{b}_-(p)^{-1} \ti{b}_+(p), \qquad \pi(p) > \zeta(n/t),\\
\ti{B}_-(p)^{-1} \ti{B}_+(p), \qquad \pi(p) < \zeta(n/t),\\
\end{cases}
\ee
where
\[
\ti{b}_-(p) = \begin{pmatrix} 1 & \frac{T(p)}{T(p^*)} R(p^*) \Theta(p^*) \E^{-t\phi(p)} \\ 0 & 1 \end{pmatrix},\,
\ti{b}_+(p) = \begin{pmatrix} 1 & 0 \\ \frac{T(p^*)}{T(p)} R(p) \Theta(p) \E^{t\phi(p)} & 1 \end{pmatrix},
\]
and
\[
\ti{B}_-(p) = \begin{pmatrix} 1 & 0 \\ - \frac{T_-(p^*)}{T_-(p)} R(p) \Theta(p) \E^{t\phi(p)}& 1 \end{pmatrix}\!, \,
\ti{B}_+(p) = \begin{pmatrix} 1 & - \frac{T_+(p)}{T_+(p^*)} R(p^*) \Theta(p^*) \E^{-t\phi(p)} \\ 0 & 1 \end{pmatrix}\!.
\]
In our next step we make a contour deformation and move the corresponding parts into regions where
the off-diagonal terms are exponentially decreasing. For this we take some sufficiently small loops $C_k$ around
each spectral band $[E_{2k},E_{2k+1}]$ on the upper sheet (see Figure~\ref{fig3}). In particular, these loops must not intersect with any of 
the loops around the eigenvalues $\rho_j$.
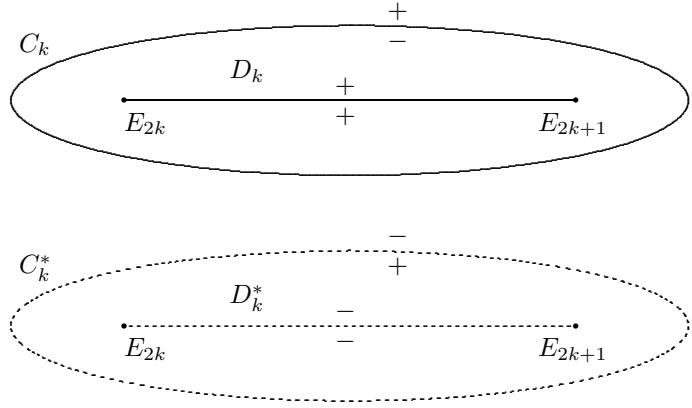
\begin{figure}
\begin{picture}(10,6)

\put(3.4,4.8){$D_k$}
\put(0.6,5.15){$C_k$}
\put(2,4.5){\circle*{0.06}}
\put(2,4.1){$E_{2k}$}
\put(8,4.5){\circle*{0.06}}
\put(7.5,4.1){$E_{2k+1}$}

\put(4.8,4.6){$+$}
\put(4.8,4.2){$+$}

\put(5.5,5.6){$+$}
\put(5.5,5.2){$-$}

\put(2,4.5){\line(1,0){6}}
\put(5,4.5){\curve(0, 1, 2.645, 0.809, 4.28, 0.309, 4.28, -0.309, 2.645, -0.809, 0, -1, %
-2.645, -0.809, -4.28, -0.309, -4.28, 0.309, -2.645, 0.809, 0, 1)}


\put(3.4,1.8){$D_k^*$}
\put(0.6,2.2){$C_k^*$}
\put(2,1.5){\circle*{0.06}}
\put(2,1.1){$E_{2k}$}
\put(8,1.5){\circle*{0.06}}
\put(7.5,1.1){$E_{2k+1}$}

\put(4.8,1.6){$-$}
\put(4.8,1.2){$-$}

\put(5.5,2.6){$-$}
\put(5.5,2.2){$+$}

\curvedashes{0.05,0.05}
\put(2,1.5){\curve(0,0,6,0)}
\put(5,1.5){\curve(0, 1, 2.645, 0.809, 4.28, 0.309, 4.28, -0.309, 2.645, -0.809, 0, -1, %
-2.645, -0.809, -4.28, -0.309, -4.28, 0.309, -2.645, 0.809, 0, 1)}
\end{picture}
\caption{The small lens contour around a spectral band.
Views from the top and bottom sheet.} \label{fig3}
\end{figure}
Then, an investigation of the sign of $\re(\phi)$ shows that
\be
\begin{cases}
\re(\phi(p) ) < 0, & p \in D_k, \: \pi(p) > \zeta(n/t)\\
 \re(\phi(p) ) > 0 ,& p \in D_k, \: \pi(p) < \zeta(n/t)
\end{cases}
\ee
and we can deform our contour according to
\be
\hat{m}(p) = \begin{cases}
\ti{m}(p) \ti{b}_+(p)^{-1}, & p \in D_k, \: \pi(p) > \zeta(n/t),\\
\ti{m}(p) \ti{b}_-(p)^{-1}, & p \in D_k^*, \: \pi(p) > \zeta(n/t),\\
\ti{m}(p) \ti{B}_+(p)^{-1}, & p \in D_k, \: \pi(p) < \zeta(n/t),\\
\ti{m}(p) \ti{B}_-(p)^{-1}, & p \in D_k^*, \: \pi(p) < \zeta(n/t),\\
\ti{m}(p), & \text{else},
\end{cases}
\ee
such that the jump on $\Sigma$ disappears and on $\bigcup_{k=0}^g (C_k \cup C_k^*)$ is given by
\be
\hat{v}(p) = \begin{cases}
\ti{b}_+(p), & p \in C_k, \: \pi(p) > \zeta(n/t),\\
\ti{b}_-(p)^{-1}, & p \in C_k^*, \: \pi(p) > \zeta(n/t),\\
\ti{B}_+(p), & p \in C_k, \: \pi(p) < \zeta(n/t),\\
\ti{B}_-(p)^{-1}, & p \in C_k^*, \: \pi(p) < \zeta(n/t),\\
\ti{v}(p), & \text{else}.
\end{cases}
\ee
The jumps on the small circles around the eigenvalues remain unchanged.

Here we have assumed that $R$ has an analytic extension to the corresponding regions. We will show
how this restriction can be overcome below. Moreover, in order to obtain uniform errors we need to require
that the $\nu_j$ retain some minimal distance from the jump contour. But this implies that some regions
of $(n,t)$ values are not covered. However, we can evade this obstacle by slightly deforming our contour such
that it has some positive distance to the first one. This way we also cover the missing regions.

Now we are ready to prove Theorem~\ref{thm:asym} by applying Theorem~\ref{thm:remcontour} as follows:

If $|\zeta(n/t) - \rho_k|>\eps$ for all $k$ we can choose $\gam_0=0$ and $w_0^t$ by removing all jumps
corresponding to poles from $w^t$. In particular, the error between the solutions of $w^t$ and $w_0^t$
is exponentially small. This proves the second part of Theorem~\ref{thm:asym} upon comparing
\be
m(p) = \hat{m}(p) \begin{pmatrix} T(p^*,n,t)^{-1} & 0\\ 0 & T(p,n,t)^{-1} \end{pmatrix}
\ee
with \eqref{m2infp} using \eqref{Tinfp}. 

Otherwise, if $|\zeta(n/t) - \rho_k|<\eps$ for some $k$, we choose $\gam_0^t=\gam_k(n,t)$ and $w_0^t \equiv 0$.
Again we conclude that the error between the solutions of $w^t$ and $w_0^t$ is exponentially small.

If $R(p)$ has no analytic extension, we will approximate $R(p)$ by analytic functions in the spirit of \cite{dz}. In fact, as in
\cite[Sect.~5]{krt2} one sees that it indeed suffices to find an analytic approximation for the left and right reflection
coefficients. Moreover, for each spectral band (viewed as a circle on the Riemann surface)
one can take the imaginary part of the phase as a coordinate transform and then use the
usual Fourier transform with respect to this coordinate (compare \cite[Lem.~5.3]{krt2}). In order to
avoid problems when one of the poles $\nu_j$ hits $\Sigma$, one just has to make the approximation in
such a way that the nonanalytic rest vanishes at the band edges (cf.\ Remark~\ref{rem:sieqpole}).
That is, split $R$ according to
\begin{align} \nn
R(p) =& R(E_{2j}) \frac{z-E_{2j}}{E_{2j+1}-E_{2j}} + R(E_{2j+1}) \frac{z-E_{2j+1}}{E_{2j}-E_{2j+1}}\\
& \pm \sqrt{z-E_{2j}} \sqrt{z-E_{2j+1}} \ti{R}(p), \qquad p=(z,\pm),
\end{align}
and approximate $\ti{R}$.

\begin{remark}
Note that one can even do slightly better by using the weighted measure $-\I\Rg{p}d\pi$ on $\Sigma$,
in which case it suffices if $R$ is just $C^{l+1,\gam}(\Sigma)$ for some $\gam>0$ rather than $C^{l+2}(\Sigma)$.
In fact, one can show that the Cauchy operators are still bounded in this weighted Hilbert space
(cf.\ \cite[Thm.~4.1]{gk}).
\end{remark}

\noindent
{\bf Acknowledgments.}
G.T. gratefully acknowledges the extraordinary hospitality of the Department of Mathematics at
Rice University, where part of this research was done.

\appendix

\section{Singular integral equations}
\label{sec:sieq}

In this section we show how to transform a meromorphic vector Riemann--Hilbert problem
with simple poles at $\rho$, $\rho^{\ast}$,
\begin{align}\nn
& m_+(p) = m_-(p) v(p), \qquad p\in \Sigma,\\ \nn
& (m_1) \ge -\dimuzs - \dirho, \qquad (m_2) \ge -\dimuz - \dirhos,\\ \label{eq:rhp5m}
& \Big( m_1(p) - \frac{R^{1/2}_{2g+2}(\rho)}{\prod_{k=1}^g (\rho-\mu_k)}
\frac{\gam_j}{\pi(p)-\rho} \frac{\psi_q(p)}{\psi_q(p^*)} m_2(p) \Big) \ge - \dimuz^*,
\mbox{ near $\rho$},\\ \nn
& \Big( - \frac{R^{1/2}_{2g+2}(\rho)}{\prod_{k=1}^g (\rho-\mu_k)}
\frac{\gam}{\pi(p)-\rho} \frac{\psi_q(p)}{\psi_q(p^*)}  m_1(p) + m_2(p) \Big) \ge - \dimuz,
\mbox{ near $\rho^*$},\\ \nn
& m(p^*) = m(p) \sigI,\\ \nn
& m(\infty_+) = \begin{pmatrix} 1 & m_2\end{pmatrix},
\end{align}
into a singular integral equation.
Since we require the symmetry condition \eqref{eq:symcond} for our Riemann--Hilbert
problems, we need to adapt the usual Cauchy kernel to preserve this symmetry.
Moreover, we keep the single soliton as an inhomogeneous term which will play
the role of the leading asymptotics in our applications.

\begin{hypothesis}\label{hyp:sym}
Let $\Sigma$ consist of a finite number of smooth oriented finite curves in $\mathbb{M}$
which intersect at most finitely many times with all intersections being transversal.
Assume that the contour $\Sigma$ does not contain $\infty_\pm$ and is invariant under
$p\mapsto p^{\ast}$. It is oriented such that under the mapping $p\mapsto p^{\ast}$
sequences converging from the positive sided to $\Sigma$ are mapped to sequences
converging to the negative side. Moreover, suppose the jump matrix $v$
can be factorized according to $v = b_-^{-1} b_+ = (\id-w_-)^{-1}(\id+w_+)$, where
$w_\pm = \pm(b_\pm-\id)$ are continuous and satisfy
\be\label{eq:wpmsym}
w_\pm(p^\ast) = \sigI w_\mp(p) \sigI,\quad z\in\Sigma.
\ee
\end{hypothesis}

In order to respect the symmetry condition we will restrict our attention to
the set $L^2_{s}(\Sigma)$ of square integrable functions $f:\Sigma\to\C^{2}$ such that
\be\label{eq:sym}
f(p) = f(p^\ast) \sigI.
\ee
Clearly this will only be possible if we require our jump data to be symmetric as well (i.e.,
Hypothesis~\ref{hyp:sym} holds).

We begin by introducing the Cauchy operator following
Section~5 in \cite{kt2}. Given a nonspecial divisor $\dimuz$,
we introduce the Cauchy kernel
\be\label{defOm}
\Omega_p^{\uhmuz,\rho} = \om_{p\, \rho} + \sum_{j=1}^g I_j^{\uhmuz,\rho}(p) \zeta_j,
\ee
where
\be
I_j^{\uhmuz,\rho}(p) = \sum_{\ell=1}^g c_{j\ell}(\uhmuz) \int_\rho^p \om_{\hmu_\ell,0}.
\ee
Here $\om_{q,0}$ is the (normalized) Abelian differential of the second kind with
a second order pole at $q$.
Note that $I_j^{\uhmuz,\rho}(p)$ has first order poles at the points $\uhmuz$.

Introduce
\be
\ul{\Omega}_p^{\uhmuz,\rho} = \begin{pmatrix} \Omega_{p}^{\uhmuz^*,\rho^*} & 0 \\                                       
0 & \Omega_p^{\uhmuz,\rho} \end{pmatrix},
\ee
and define the Cauchy operator by
\be
(C f)(p) = \frac{1}{2\pi\I} \int_\Sigma f(s) \ul{\Omega}_p^{\uhmuz,\rho}
\ee
acting on vector-valued functions $f:\Sigma\to\C^{2}$. We will assume that $\dimuz$ does
not hit $\Sigma$ (see Remark~\ref{rem:sieqpole} below for the case where this assumptions does
not hold).

\begin{lemma}
Assume Hypothesis~\ref{hyp:sym}.
The Cauchy operator $C$ has the properties, that the non-tangential boundary limits
\be
(C_\pm f)(q) = \lim_{p \to q \in \Sigma} \frac{1}{2 \pi \I} \int_{\Sigma} f\, \ul{\Omega}_p^\uhnuz
\ee
from the left and right of $\Sigma$ respectively (with respect to its orientation)
are bounded operators $L^2_s(\Sigma) \to L^2_s(\Sigma)$. The bound
can be chosen independent of the divisor as long as it stays some finite
distance away from $\Sigma$. The operators $C_\pm$ satisfy
\be\label{eq:cpcm}
C_+ - C_- = \id
\ee
and
\be\label{eq:Cnorm}
(Cf)(\rho^\ast) = (0\quad\ast), \qquad (Cf)(\rho) = (\ast\quad 0).
\ee
Furthermore, $C$ restricts to $L^2_{s}(\Sigma)$, that is
\be
(C f) (p^\ast) = (Cf)(p) \sigI,\quad p\in\mathbb{M}\backslash\Sigma
\ee
for $f\in L^2_{s}(\Sigma)$ and if $w_\pm$ satisfy \eqref{eq:wpmsym} we also have
\be \label{eq:symcpm}
C_\pm(f w_\mp)(p^\ast) = C_\mp(f w_\pm)(p) \sigI,\quad p\in\Sigma.
\ee
\end{lemma}

\begin{proof}
Follows from the properties of Cauchy operators, see Theorem~5.1 in \cite{kt2}.
\end{proof}

We have thus obtained a Cauchy transform with the required properties.
Following Sections~7 and 8 of \cite{bc}, we can solve our Riemann--Hilbert problem using this
Cauchy operator.

Introduce the operator $C_w: L_s^2(\Sigma)\to L_s^2(\Sigma)$ by
\be
C_w f = C_+(f w_{-}) + C_-(f w_{+}),\quad f\in L^2_s(\Sigma).
\ee
Recall from Lemma~\ref{lem:singlesoliton} that the
unique solution corresponding to $v\equiv \id$ is given by
\[
m_0(p)= \begin{pmatrix} f(p^\ast) & f(p) \end{pmatrix}, 
\]
for some given $f(p)$ with $(f) \geq - \di_{\uhmuz} - \di_{\rho^*}$.
Since we assumed $\dimuz$ to be away from $\Sigma$, we clearly have $m_0 \in L^2(\Sigma)$.

\begin{theorem}\label{thm:cauchyop}
Assume Hypothesis~\ref{hyp:sym}.

Suppose $m$ solves the Riemann--Hilbert problem \eqref{eq:rhp5m}. Then
\be\label{eq:mOm}
m(p) = (1-c_0) m_0(p) + \frac{1}{2\pi\I} \int_{\Sigma} \mu(s) (w_{+}(s) + w_{-}(s)) \ul{\Omega}_{p}^{\uhmuz,\rho},
\ee
where
\[
\mu = m_+ b_+^{-1} = m_- b_-^{-1} \quad\mbox{and}\quad
c_0= \left( \frac{1}{2\pi\I} \int_{\Sigma} \mu(s) (w_{+}(s) + w_{-}(s)) \ul{\Omega}_{\infty_+}^{\uhmuz,\rho} \right)_{\!1}.
\]
Here $(m)_j$ denotes the $j$'th component of a vector.
Furthermore, $\mu$ solves
\be\label{eq:sing4muc}
(\id - C_w) \mu = (1-c_0) m_0(p).
\ee

Conversely, suppose $\ti{\mu}$ solves 
\be\label{eq:sing4mu}
(\id - C_w) \ti{\mu} = m_0,
\ee
and
\[
\ti{c}_0= \left( \frac{1}{2\pi\I} \int_{\Sigma} \ti{\mu}(s) (w_{+}(s) + w_{-}(s)) \ul{\Omega}_{\infty_+}^{\uhmuz,\rho} \right)_{\!1} \ne 0,
\]
then $m$ defined via \eqref{eq:mOm}, with $(1-c_0)=(1-\ti{c}_0)^{-1}$ and $\mu=(1-\ti{c}_0)^{-1}\ti{\mu}$,
solves the Riemann--Hilbert problem \eqref{eq:rhp5m} and $\mu= m_\pm b_\pm^{-1}$.
\end{theorem}

\begin{proof}
Follows as in \cite{krt}.
\end{proof}

Hence we have a formula for the solution of our Riemann--Hilbert problem $m(z)$ in terms of
$(\id - C_w)^{-1} m_0$ and this clearly raises the question of bounded
invertibility of $\id - C_w$. This follows from Fredholm theory (cf.\ e.g. \cite{zh}):

\begin{lemma}
Assume Hypothesis~\ref{hyp:sym}.

The operator $\id-C_w$ is Fredholm of index zero,
\be
\ind(\id-C_w) =0.
\ee
\end{lemma}

\begin{proof}
Follows as in \cite{kt2}, \cite{krt}.
\end{proof}

By the Fredholm alternative, it follows that to show the bounded invertibility of $\id-C_w$
we only need to show that $\ker (\id-C_w) =0$. The latter being equivalent to
unique solvability of the corresponding vanishing Riemann--Hilbert problem.

\begin{corollary}
Assume Hypothesis~\ref{hyp:sym}.

A unique solution of the Riemann--Hilbert problem \eqref{eq:rhp5m} exists if and only if the corresponding
vanishing Riemann--Hilbert problem, where the normalization condition is replaced by
$m(0)= \begin{pmatrix} 0 & m_2\end{pmatrix}$, has at most one solution.
\end{corollary}

We are interested in comparing two Riemann--Hilbert problems associated with
respective jumps $w_0$ and $w$ with $\|w-w_0\|_\infty$ small,
where
\be
\|w\|_\infty= \|w_+\|_{L^\infty(\Sigma)} + \|w_-\|_{L^\infty(\Sigma)}.
\ee
For such a situation we have the following result:

\begin{theorem}\label{thm:remcontour}
Assume that for some data $w_0^t$ the operator
\be
\id-C_{w_0^t}: L^2_s(\widehat{\Sigma}) \to L^2_s(\widehat{\Sigma})
\ee
has a bounded inverse, where the bound is independent of $t$, and
let $\zeta=\zeta_0$, $\gam^t=\gam_0^t$.

Furthermore, assume $w^t$ satisfies
\be
\|w^t - w_0^t\|_\infty \leq \alpha(t)
\ee
for some function $\alpha(t) \to 0$ as $t\to\infty$. Then
$(\id-C_{w^t})^{-1}: L^2_s(\Sigma)\to L^2_s(\Sigma)$ also exists
for sufficiently large $t$ and the associated solutions of
the Riemann--Hilbert problems \eqref{eq:rhp5m} only differ by $O(\alpha(t))$.
\end{theorem}

\begin{proof}
Follows as in \cite{krt}.
\end{proof}

\begin{remark}\label{rem:sieqpole}
The case where one (or more) of the poles $\hat\mu_j$ lies on $\Sigma$ can be included if one
assumes that $w_\pm$ has a first order zero at $\hat\mu_j$. In fact, in this case one can replace
$\mu(s)$ by $\ti{\mu}(s)=(\pi(s)-\mu_j)\mu(s)$ and $w_\pm(s)$ by $\ti{w}_\pm(s)=(\pi(s)-\mu_j)^{-1}w_\pm(s)$.
\end{remark}

\noindent
{\bf Acknowledgments.}
We thank A.\ Mikikits-Leitner for pointing out errors in a previous version of this article.
G.T. gratefully acknowledges the extraordinary hospitality of the Department of Mathematics at Rice University,
where part of this research was done.

\end{document}